\documentclass{article}

\usepackage{arxiv}
\usepackage{amsthm}
\usepackage{caption}
\usepackage{graphicx} 
\usepackage{subcaption}
\usepackage{amsmath}
\usepackage{amsfonts}
\usepackage{array}
\usepackage{mathrsfs}
\usepackage{mdframed}
\usepackage{amssymb}
\usepackage{fancyhdr}
\usepackage{changepage}
\usepackage{url}
\usepackage{algpseudocode}
\usepackage{algorithm}
\usepackage[english]{babel}
\newtheorem{theorem}{Theorem}

\usepackage{textcomp}
\usepackage{textalpha}

\usepackage{url}

\newcommand{\chat}{\hat{\mathbf{c}}}
\newcommand{\normc}{\left\langle \mathbf{c} \right\rangle}
\newcommand{\normchat}{\left\langle \hat{\mathbf{c}} \right\rangle}

\newcommand{\knnq}{knn(\mathbf{q})}
\newcommand{\knnmu}{knn(\boldsymbol{\mu})}
\newcommand{\cnobold}{c}
\newcommand{\chatnobold}{\hat{c}}
\newcommand{\eqdef}{\overset{def}{=}}
\usepackage{booktabs}

\DeclareMathOperator*{\argmax}{argmax}
\DeclareMathOperator*{\argmin}{argmin}
\algdef{SE}[DOWHILE]{Do}{doWhile}{\algorithmicdo}[1]{\algorithmicwhile\ #1}%
\title{Dynamic Query Modification for Binary Locality Sensitive Hashing}
\author{
  Ben Claydon, Richard Connor, Alan Dearle \\
  University of St Andrews \\
  St Andrews, Scotland\\
  \texttt{\{bc89,rchc,al\}@st-andrews.ac.uk} \\   
}
\date{May 2026}

\begin{document}

\maketitle

\begin{abstract}
    Our context of interest is how binary locality sensitive hash (LSH) functions can be used to solve the approximate near neighbour (ANN) problem in hyperspherical spaces, which seeks to find the $k$ closest elements of some dataset $X \subset \mathbb{R}^d$ to some further point $\mathbf{q} \in \mathbb{R}^d$ presented as a query. Binary locality sensitive function families $\mathcal{H}$ are sets of functions each with signature $h: \mathbf{x} \in \mathbb{R}^d \rightarrow \{ 0, 1 \}$. A function is locality sensitive if and only if the output of the function is more likely to be equal (a `hash collision') if two close vectors are used as input than if two far vectors are used. 
    
    A data structure can be built by generating binary hash codes for each member of $X$, which are generated by drawing and applying one or more functions from $\mathcal{H}$. When $\mathbf{q}$ is presented as a query, the same set of functions are applied to it and those elements of $X$ with equal binary hash codes retrieved.

    In this paper we introduce \emph{dynamic query modification}. This procedure changes $\mathbf{q}$ at query time using information acquired during search to approximate the point $\normc$, which is the $\ell_2$-normalised centroid of the $k$-nearest neighbour set of $\mathbf{q}$. By theoretical and experimental analysis we demonstrate $\normc$ has two significant advantages when compared to $\mathbf{q}$. Firstly, we show that the hash function output of $\normc$ is likely to have a higher average collision probability with the $k$-nearest neighbours than the hash output of $\mathbf{q}$. Secondly, in our experiments, we observe no instances of $\normc$ failing to collide with any member of the near-neighbour set; a property which we demonstrate is not true for $\mathbf{q}$.

    To demonstrate the efficacy of the technique, we define a novel structure MQ-Forest, a modified version of RP-Forest. Both are binary LSH-based ANN mechanisms, but MQ-Forest dynamically estimates $\normc$ during the query process. We show that MQ-Forest reduces both build and query times by up to $40\%$ when measured over several large, high-dimensional benchmark datasets.
\end{abstract}

\section{Introduction}

Searching collections of large, high-dimensional data is an increasingly important task in modern applications. Our context of interest is finding an approximation of the $k$ closest elements of some dataset $X \subset \mathbb{R}^d$ to some further point $\mathbf{q} \in \mathbb{R}^d$ presented as a query.  We exclusively consider hyperspherical spaces, which are those where all vectors have an $\ell_2$-norm of $1$. Finding the exact $k$ nearest neighbours to $\mathbf{q}$ in high-dimensional datasets is known to require a similarity comparison between $\mathbf{q}$ and each member of $X$ due to the `curse of dimensionality'. However, scalable solutions to such queries are required by many applications. Approximate near neighbour (ANN) solutions allow search in sublinear time.

One technique used by algorithms in this category is locality sensitive hashing (LSH). Binary LSH functions accept a point $\mathbf{x} \in \mathbb{R}^d$ as input, and return a binary output value. They possess the property that the output of the function when applied to two near points is more likely to be equal (a `hash collision') than when the function is applied to two far points.

It may seem intuitive that the query sits at the centre of $\knnq$. We show that this is not typically the case which introduces significant pathologies for LSH-based mechanisms. This suggests the surprising possibility that there exists an input to LSH functions which more effectively retrieves $\knnq$.

In this paper, we describe a significant improvement which can be made to LSH-based ANN algorithms based on some novel mathematical observations on high-dimensional spaces. We make the following observations:

\begin{itemize}
    \item In Section \ref{sec:theory}, we demonstrate the existence of a set $\Phi \subset \mathbb{R}^d$. Each member of $\Phi$ has the property that, when used as input to a locality sensitive hash function, they have an equal or higher probability than $\mathbf{q}$ of hash colliding with near neighbours of $\mathbf{q}$. If an element of $\Phi$ can be identified and queried in place of $\mathbf{q}$, the performance of LSH-based ANN algorithms can be improved as near neighbours of $\mathbf{q}$ will be retrieved with a greater frequency.

    \item In Section \ref{sec:application_to_forest}, we provide a practical and fast method to compute a vector likely to be a member of $\Phi$ at query time, which can be used to accelerate the performance of binary LSH-based mechanisms.
        
    \item In Section \ref{sec:experiments}, we apply the methods described in Section \ref{sec:application_to_forest} to a popular binary LSH-based ANN algorithm: RP-Forest. We validate the properties of our query modification technique by contrasting the behaviour of a modified and unmodified version of RP-Forest over large, high-dimensional benchmark datasets. We show that our modification to the algorithm results in reduced build and query times when compared to a baseline implementation. 
\end{itemize}

\section{Related Work}

\subsection{Locality Sensitive Hashing}
\label{sec:lsh}
LSH was introduced by Indyk and Motwani \cite{Indyk_Motwani_1998} as a method for searching high-dimensional spaces. A family of functions $\mathcal{H}$ is locality sensitive if and only if there exists a $p_1 > p_2$ and $t_1 < t_2$ which satisfy inequalities \ref{eq:threshold_p1} and \ref{eq:threshold_p2} for all functions $h \in \mathcal{H}$. A corollary of this property is that for any two points $\mathbf{x}, \mathbf{y}$, the probability $\Pr(h(\mathbf{x}) = h(\mathbf{y}))$ decreases monotonically as $||\mathbf{x} - \mathbf{y}||_2$ increases.

\begin{alignat}{2}
    \Pr(h(\mathbf{x}) = h(\mathbf{y})) &\geq p_1 \quad &&| \quad ||\mathbf{x} - \mathbf{y}||_2 \leq t_1 \label{eq:threshold_p1} \\
    \Pr(h(\mathbf{x}) = h(\mathbf{y})) &\leq p_2 \quad &&| \quad ||\mathbf{x} - \mathbf{y}||_2 > t_2 \label{eq:threshold_p2}
\end{alignat}

Charikar introduced a binary LSH function for Euclidean spaces using randomly generated hyperplanes which pass through the origin \cite{Charikar_2002}. A hash value of $1$ is assigned to a point $\mathbf{x}$  if it is above the hyperplane, or $0$ otherwise. This LSH function family is almost identical to the node subdivision procedure found in Random Projection Forests (see Section \ref{sec:rp_tree}); the latter only differs from Charikar's original formulation by shifting the hyperplane a random amount from the origin.

\subsubsection{LSH Approximate Near Neighbour Algorithms}
LSH functions are often used to create data structures which solve the ANN problem. Locality-sensitivity based ANN algorithms can be viewed abstractly in a manner described by Lv \cite{Lv_Josephson_Wang_Charikar_Li}:

\begin{quote}
    The basic LSH indexing method processes a similarity search, for a given query q, in two steps. The first step is to generate a candidate set by the union of all buckets that query q is hashed to. The second step ranks the objects in the candidate set according to their distances to query object q, and then returns the top K objects. \\
    $\cdots$\\
    By concatenating multiple LSH functions, the collision probability of far away objects becomes very small ($p_2$), but it also reduces the collision probability of nearby objects ($p_1$). As a result, multiple hash tables are needed in order to find most of the nearby objects.
\end{quote}

\subsection{RP-Trees and RP-Forests}
\label{sec:rp_tree}

The RP-Tree data structure was first introduced by Dasgupta and Freund as a variant of the KD-Tree suitable for searching high-dimensional spaces \cite{rptree}. To build an RP-Tree, a node is initialised which forms the root of the tree, and all points in a dataset $X \subset \mathbb{R}^d$ are assigned to this node. Thereafter, a tree is constructed by repeatedly subdividing any node containing a number of points greater than some user-defined parameter $n_s$ into two child nodes.

To subdivide a node, any point $\mathbf{x}$ contained in the node such that $\mathbf{w} \cdot \mathbf{x} > a$, for some plane $\mathbf{w}$ and some real value $a$, are assigned to the left child, and all other points to the right child. The value $a$ is uniformly randomly generated with upper and lower bounds equal to the greatest height above and below $\mathbf{w}$ of any point contained in the node. When $a=0$, this hash function is analogous to Charikar's hyperplane LSH.

The query procedure of an RP-Tree begins at the root node, where the predicate $\mathbf{w} \cdot \mathbf{q} > a$ is computed using the same instance of $\mathbf{w}$ and $a$ which split the root node. If this predicate is true, then the left subtree is recursively searched, or the right subtree is recursively searched if not. This process continues until a leaf node is encountered, and the objects contained therein are retrieved as a `candidate set'. Only points in the candidate set will be measured for similarity to $\mathbf{q}$, and the $k$ closest returned as the approximate near neighbour set.

As this paper discusses two separate families of binary LSH function, we refer to $\mathcal{H}_{rp}$ as the LSH schema used by RP-Forest, and $\mathcal{H}_{c}$ as the LSH schema defined by Charikar. We use $\mathcal{H}$ to refer to either function family.

RP-Forest is an ensemble of several independently constructed RP-Trees. When a query is presented, each RP-Tree is searched, and a candidate set is formed by taking the union of all objects retrieved from the trees. We use Yan et al.'s formulation of RP-Forest, laid out in Algorithm 1 of their study \cite{Yan_Wang_Wang_Wang_Li_2018}. Applications include studies on gene data \cite{Gondara_2015} \cite{Zhang_2011}, pose estimation \cite{Lee_Yang_Oh_2015} \cite{Lee_Yang_Oh_2019}, and outlier detection \cite{Tan_Yang_Rahardja_2022}.

\subsection{Query Modification}

Here we review methods of modifying a user-provided query vector to permit the accelerated retrieval of relevant elements from a large collection, which is also the aim of the novel query modification algorithm we present.

\subsubsection{Multiprobe LSH}
\label{sec:multiprobe}
There are two forms of multiprobe LSH: entropy-based \cite{Lv_Josephson_Wang_Charikar_Li} and those which are dependent on LSH bin geometry such as the schema introduced by Andoni et al. for the cross-polytope LSH family \cite{Andoni_Indyk_Laarhoven_Razenshteyn_Schmidt_2015}. We will focus on the former, as it is a technique which relies on truly modifying the query vector rather than an additional step in an LSH-based query mechanism.

Entropy-based multiprobe is the process of perturbing a query vector by a random Gaussian vector in order to change its hash value. If the degree of perturbation is scaled appropriately, then a new vector is generated which is similar to the original query but hashes to a different value. Additional candidates may be retrieved from this new hash bin, as it is also likely to contain items which are similar to $\mathbf{q}$. In this scenario a new hash table need not be searched to find additional relevant objects, which reduces the amount of space and time required by the LSH-based ANN algorithm.

The methods we present also rely on modifying the query to raise the likelihood of retrieving similar objects to the query by identifying and using a point in $\Phi$. However, we exploit information gained during a search procedure, as opposed to an entirely random perturbation. In this way, our query modification method is complementary to multiprobe LSH as both techniques can be used simultaneously. 

\subsubsection{Relevance Feedback}

The idea of modifying a query to permit the retrieval of more relevant objects from a large collection is well established as `relevance feedback' in the domain of information retrieval. Relevance feedback attempts to solve the problem of different users of an information retrieval (IR) system having different needs. Results retrieved by an IR system may be relevant to some users, but not others depending on the requirements of their query.

To solve this problem, users are asked to provide relevance feedback on each member of the set of results retrieved by the IR system. This feedback can take the form of the binary label `relevant' or `not relevant', or some numerical score of relevance \cite{Zhou_Huang_2003}. Upon receiving this feedback, the IR system will formulate a new query, and perform another search using it. This new query should somehow be more similar to the `relevant' documents, and dissimilar to the `irrelevant' documents.

One classical example of a method to modify a query vector using relevance feedback is Rocchio's algorithm \cite{rocchio1971relevance}. This algorithm modifies the query vector in the following way:

\begin{equation}
    \mathbf{u} = a \mathbf{q} + b \frac{1}{D_r} \sum_{\mathbf{r} \in D_r} \mathbf{r} - c \frac{1}{D_{nr}} \sum_{\mathbf{n} \in D_{nr}} \mathbf{n}
\end{equation}

where $a, b, c$ are some weights, $D_r$ is the set of relevant items as judged by the user, and $D_{nr}$ are the irrelevant items. Rocchio's algorithm is conceptually simple: it moves the query vector closer to the centroid of the set of relevant objects, and away from the centroid of irrelevant objects. 

Rocchio's algorithm is applied iteratively by generating several sets of results and incorporating new user feedback for each. Examples of systems which use relevance feedback in this way are \cite{Vadicamo_Scotti_Dearle_Connor_2025} and \cite{Claydon_Connor_Dearle_Vadicamo_2025}. This iterative process will result in a query which is personalised to the user's specific requirements.

\subsubsection{Pseudo-Relevance Feedback}
Relevance feedback assumes that there is a human in the loop, who guides the query as it progresses. Some IR systems utilise \emph{pseudo-relevance feedback} (PRF), which assumes the most similar documents to the query are `relevant' and others are `not relevant'. This technique permits the refining of queries automatically, without additional human input.

Examples of use of this technique include local context analysis \cite{Xu_Croft_2000}, which is a method of applying PRF to document retrieval. After results are retrieved by the IR system, local context analysis automatically forms a new `expanded' query by selecting new keywords based on those which frequently occur in retrieved documents. An expanded query is a query which has been modified to include additional relevant information. Results show that this expanded query retrieves more relevant documents than the original.

This technique can also be used in dense vector spaces, such as the ones we study. Kuzi et al. studied the retrieval of documents using expanded queries \cite{word_expansion}. Instead of an iterative method, queries were expanded using similar words as judged by a Word2Vec \cite{Mikolov_Chen_Corrado_Dean_2013} word embedding model trained on the document corpus. This form of PRF was also shown to produce a query which improves retrieval of relevant documents.

Faggioli et al \cite{Faggioli_Ferro_Perego_Tonellotto_2024} presented a learning-based approach to form an expanded query by selecting some small number of important dimensions of the original query embedding, and setting the rest to 0. By modifying the query in this way, the authors showed that more relevant documents are retrieved than when the original query is used. Furthermore, query expansion methods for LSH algorithms have been attempted \cite{Cui_Li_Zhu_Li_Zhang_2024} \cite{Huang_Feng_Zhang_Fang_Ng_2015} \cite{Kuo_Chen_Chiang_Hsu_2009} with learning-based approaches. 

\section{Properties of Generic LSH Mechanisms}

Any family of locality sensitive functions must obey Inequalities \ref{eq:threshold_p1} and \ref{eq:threshold_p2}. In the context of the ANN problem, the property of interest is \emph{what fraction} of $\knnq$ collide with $\mathbf{q}$. However, the properties of an LS function family make no such guarantees in the latter context.

In previous work, we have shown how this lack of guarantee can hinder binary LSH-based ANN mechanisms significantly \cite{sisap25}. There, we introduced the notion of a `hash failure', and demonstrated they occur frequently. A hash failure occurs when $\mathbf{q}$ hash collides with only a small fraction of its near neighbours. A hash failure can occur because, even if the probability of hash collision between $\mathbf{q}$ and each of its near neighbours is individually high, these probabilities may be \emph{non-independent} and therefore strongly correlated since independence is not a property guaranteed by an LS function family.

The existence of such failure cases particularly impacts the probability of hash collision between \emph{compound} hash functions. Compound hash functions are formed from the combination of independently drawn hash functions. In the context of the binary hash functions we study, a compound hash function generates a single binary string by concatenating the output of multiple independent binary LSH functions. A hash collision between two compound hash functions occurs when these binary strings are equal. 

A compound hash function collides with a number of elements of the $k$ nearest neighbour set proportionally to the \emph{product} of the recall of each constituent binary hash function. Therefore, the failure of \emph{any} individual hash function is not tolerable.

\begin{figure}
    \centering
    \begin{subfigure}{0.45\textwidth}
        \centering
        \includegraphics[width=\linewidth]{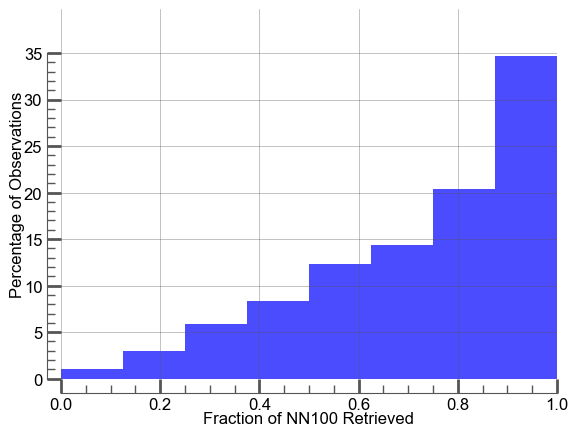}
        \caption{Using $\mathbf{q}$ as input to binary hash function}
        \label{fig:dino2_q_example_preview}
    \end{subfigure}
    \hfill
    \begin{subfigure}{0.45\textwidth}
        \centering
        \includegraphics[width=\linewidth]{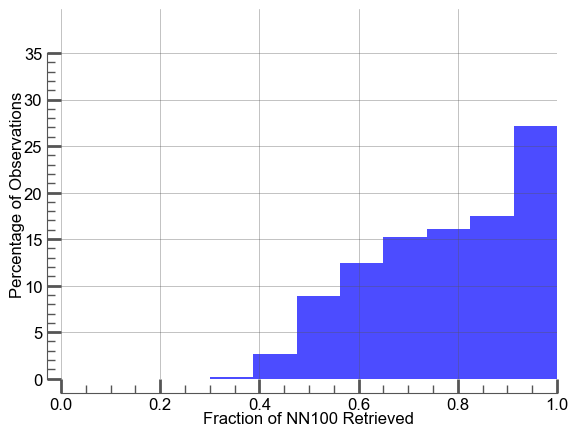}
        \caption{Using $\normc$ as input to binary hash function}
        \label{fig:dino2_c_example_preview}
    \end{subfigure}
    \caption{Fraction of points in $knn(\mathbf{q})$ which hash collide with $\mathbf{q}$ and $\normc$ respectively. Measured by applying $1000$ random functions from $\mathcal{H}_{rp}$ to $100$ random queries, over the MirFlickr Dino2 dataset. It is possible for $\mathbf{q}$ to collide with only a small fraction of $\knnq$, whereas this is not observed for an element $\normc \in \Phi$.}
    \label{fig:hash_collision_example}
\end{figure}

Figure \ref{fig:dino2_q_example_preview} shows the distribution of $\knnq$ collision rates after a large number of binary hash functions are applied to a dataset of image embeddings, with random objects from the dataset selected to use as queries. As can be seen, many functions retrieve a very small fraction of the near neighbours. Figure \ref{fig:dino2_c_example_preview} shows the same observations, where an element $\normc \in \Phi$ is used in place of $\mathbf{q}$%
\footnote{We will demonstrate how to derive $\normc$ and $\Phi$ in Section \ref{sec:query_quality}.}. A fuller explanation of this experiment, and results on other datasets, are provided in Section \ref{sec:eliminating_hash_failures}. Unlike $\mathbf{q}$, we empirically observe no instances of hash failure when this point is used as input to LSH functions.



\section{Query Modification}
\label{sec:theory}





\subsection{Defining Modified Query Quality}
\label{sec:query_quality}

We define the set $\Phi$ as those elements of $\mathbb{R}^d$ whose hashes are equally or more likely to collide with $\knnq$ than $\mathbf{q}$. We define a measure of this likelihood as the \emph{average collision probability} between some point $\mathbf{u} \in \mathbb{R}^d$ and each member of $\knnq$:

\begin{equation}
    ACP(\mathbf{u}, \knnq)
    =
    \frac{1}{k}
    \sum_{i=0}^{k-1}
    \Pr_{h \sim \mathcal{H}}
    \left[
        h(\mathbf{u}) = h(\mathbf{x}_i)
    \right],
    \qquad \mathbf{x}_i \in \knnq .
    \label{eq:average_collision_probability}
\end{equation}

We more formally define the set $\Phi$ to be those points where $ACP(\boldsymbol{\phi} \in \mathbb{R}^d \; | \; ||\boldsymbol{\phi}||_2 = 1) \geq ACP(\mathbf{q})$. In Appendix \ref{sec:statistical_model} we describe a statistical model which computes the ACP function for the family $\mathcal{H}_{rp}$, for arbitrary values of $\mathbf{u}$ and $\knnq$. Furthermore, in Section \ref{sec:proof_a_is_0}, we also show that this model can be used to compute the ACPs for the family $\mathcal{H}_c$ as well. 

We further define a member $\normc \in \Phi$ which has, to a first-order approximation, the highest ACP of any member of $\Phi$. In Appendix \ref{sec:proofs}, we demonstrate the $\ell_2$-normalised geometric centroid of the set $\knnq$ meets this condition. Moreover, this point is extremely cheap to compute given an arbitrary set $\knnq$, as we will go on to show.

\subsubsection{`Hash Failure' Cases}
\label{sec:off_diag_properties}

Maximising average collision probability is necessary but not sufficient to ensure good search outcomes. Hash failure cases must also be rare to ensure the reliable performance of compound hash functions. To prevent hash failure, we require a lack of correlation among outputs of a single hash function as applied to each element of $\knnq$. We refer to this as an \emph{average correlation} value. This is because hash failures occur when the output of a single function is equal when applied to most members of $\knnq$ due to this correlation, and $\mathbf{q}$ hashes to a different value.

In Appendix \ref{sec:off_diag_proof}, we provide a theoretical explanation of why hash failures occur, as well that hash failure cases are empirically eliminated in our experiments when $\normc$ is used as the input to hash functions. We show that the average correlation depends on only three geometric properties: the inter-point distances between members of $knn(\mathbf{q})$, the distance between $\mathbf{u}$ and each member of $knn(\mathbf{q})$, and the height of $\mathbf{u}$ above $\mathbf{w}$. As the inter-point distances between members of $knn(\mathbf{q})$ are a fixed quantity, and the height of $\mathbf{u}$ above $\mathbf{w}$ cannot be controlled as $\mathbf{w}$ is randomly generated, we may only influence the average correlation by changing the position of $\mathbf{u}$. We demonstrate that the average covariance between near neighbour heights above / below the hyperplane $\mathbf{w}$ is minimal when $\mathbf{u} = \normc$, thereby reducing the probability of most or all near neighbours being above or below $\mathbf{w}$. Furthermore, in Appendix \ref{sec:off-diag}, we experimentally verify that this value is approximately $0$ when measured over benchmark datasets.

\section{Description of LSH Algorithm Modification}
\label{sec:application_to_forest}
\subsection{Overview}
In Section \ref{sec:lsh_index_structure}, we describe a modification which can be made to binary LSH mechanisms which creates, refines, and utilises an estimate of $\normc$ at query time. We call such an estimate $\normchat$. In Section \ref{sec:fast_centroid_cal}, we discuss the computational cost of computing $\normchat$ and present a method for accelerating it.

\subsection{Structure of LSH-Based ANN Algorithms}
\label{sec:lsh_index_structure}

In the parlance of this paper, LSH-based ANN algorithms can abstractly be described as:
\begin{itemize}
    \item All LSH-based ANN mechanisms store one or more independent compound hash functions for each member of $X$. These are constructed at preprocessing time.
    \item When a query is presented, LSH-based ANN mechanisms apply the same set of compound hash functions to the query. Only those points from $X$ where at least one compound hash collision has occurred are retrieved as a candidate set.
    \item By increasing the number of compound hash functions stored, the expected recall will increase, at the cost of additional preprocessing and query time.
\end{itemize}

\begin{figure}[htbp]
    \centering
    
    \begin{subfigure}{0.8\linewidth}
        \centering
        \includegraphics[width=\linewidth]{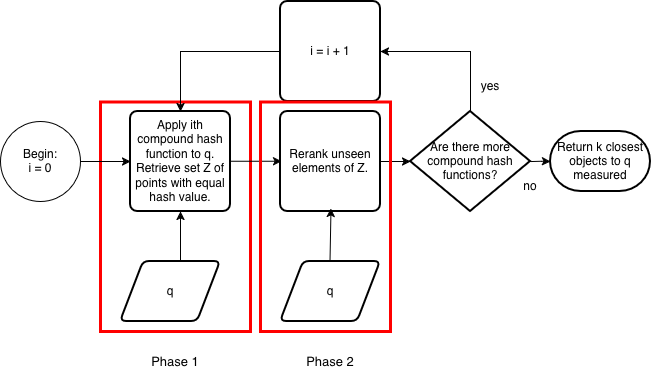}
        \caption{Abstract model of an LSH-based ANN query algorithm.}
        \label{fig:lsh_model}
    \end{subfigure}
    \hfill
    \begin{subfigure}{0.8\linewidth}
        \centering
        \includegraphics[width=\linewidth]{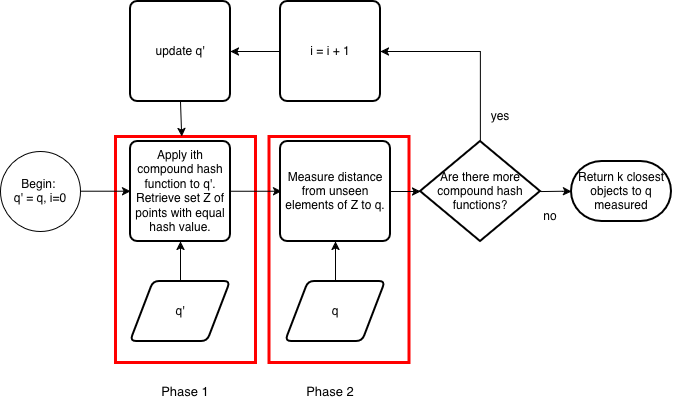}
        \caption{Abstract LSH-based ANN query algorithm which uses query updates.}
        \label{fig:disentangled_lsh}
    \end{subfigure}
    
    \caption{Comparison of query algorithm for: abstract LSH-based ANN model (above) and our modification (below).}
    \label{fig:combined_lsh}
\end{figure}

A flowchart detailing one possible implementation of this process is shown in Figure \ref{fig:lsh_model}. This abstract LSH algorithm begins by hashing the query point to create the first compound hash function. A set $Z$ is generated from those members of $X$ which have not been previously retrieved. Elements of $Z$ are re-ranked, and a candidate set is maintained of the $k$ closest to $\mathbf{q}$. This process repeats for each independent compound hash function before returning the candidate set.

Our modification to this class of algorithm is shown by the flowchart in Figure \ref{fig:disentangled_lsh}. In outline, our approach is to derive an instance of $\normchat$ from the candidate set after each compound hash function is used to retrieve new objects. This derived instance of $\normchat$ is then used as input to Phase 1 of the next compound hash function to be searched, which in turn returns new unseen points from $X$ from which a new $\normchat$ can be estimated.

We expect this derived $\normchat$ to be a member of $\Phi$ with a high probability. Thus, using $\normchat$ as input to the compound hash function instead of $\mathbf{q}$ is expected to result in both a greater probability of hash collision with near neighbours of $\mathbf{q}$, as well as a decreased probability of hash failure. For each iteration, the estimate $\normchat$ will become an increasingly accurate estimate of $\normc$.

\subsection{Cost of Query Update}
\label{sec:fast_centroid_cal}
Naively, the cost of computing the centroid of $k$ points is of $O(k)$ vector operations. As this work must be repeated each time new elements are added to the candidate set, the cost of re-computing $\normchat$ could begin to outweigh any benefit it provides. Thus, we consider how elements of the candidate set may be stored and provide an algorithm which may reduce the cost of computing its centroid. 

Information on the candidate set can be stored in a priority queue $\mathcal{Q}$ where the closest elements to $\mathbf{q}$ are given greatest priority. $\mathcal{Q}$ may contain at most $k$ elements; if a new element is added and the size of the queue would exceed $k$, the lowest priority item is removed. When an instance of $\normchat$ is required, the $\ell_2$-normalised centroid of the elements of $\mathcal{Q}$ may be computed in $k+1$ vector operations. 



We may reduce this computational cost by storing a further vector $\mathbf{s}$, which stores the value $\sum_{\mathbf{x}\in \mathcal{Q} } \mathbf{x}$. When a point $\mathbf{x}$ is added to $\mathcal{Q}$, we set $\mathbf{s} = \mathbf{s} + \mathbf{x}$. If a point $\mathbf{y}$ is removed from $\mathcal{Q}$ as a result of $\mathbf{x}$ being inserted, we set $\mathbf{s} = \mathbf{s} - \mathbf{y}$. To compute $\normchat$, we need only to divide $\mathbf{s}$ by its $\ell_2$-norm, which is an extremely fast operation. Algorithm \ref{alg:query_update} details how $\mathcal{Q}$ and $\mathbf{s}$ may be updated using a set of points $Z$, which were retrieved via hashing as in phase 2 of Figure \ref{fig:disentangled_lsh}.

\begin{algorithm}
\caption{Updating $\mathcal{Q}$ and $\mathbf{s}$}
\begin{algorithmic}[1]
    \State $\mathbf{z} \gets min(Z)$ \textit{Assign to $\mathbf{z}$ the point in $Z$ closest to $\mathbf{q}$}
    \State $\mathbf{t} \gets max(\mathcal{Q})$ \textit{Assign to $\mathbf{t}$ the point furthest from $\mathbf{q}$ in $\mathcal{Q}$}

    \While{$|\mathcal{Q}| < k$}
        \State $Z \gets Z \setminus \{\mathbf{z}\}$ \textit{Remove the closest element from $Z$}
        \State $\mathcal{Q} \gets \mathcal{Q} \cup \{\mathbf{z}\}$ \textit{Add $\mathbf{z}$ to $\mathcal{Q}$}
        \State $\mathbf{s} \gets \mathbf{s} + \mathbf{z}$ \textit{Update $\mathbf{s}$}
        \State $\mathbf{z} \gets min(Z)$
        \State $\mathbf{t} \gets max(\mathcal{Q})$
    \EndWhile

    \While{$||\mathbf{t} - \mathbf{q}||_2 > ||\mathbf{z} - \mathbf{q}||_2$}
        \State $Z \gets Z \setminus \{\mathbf{z}\}$ \textit{Remove the closest element to $\mathbf{q}$ from $Z$}
        \State $\mathcal{Q} \gets \mathcal{Q} \setminus \{\mathbf{t}\}$ \textit{Remove $\mathbf{t}$ from $\mathcal{Q}$}
        \State $\mathbf{s} \gets \mathbf{s} - \mathbf{t} + \mathbf{z}$ \textit{Update $\mathbf{s}$}
        \State $\mathcal{Q} \gets \mathcal{Q} \cup \{\mathbf{z}\}$ \textit{Add the new closest point}
        \State $\mathbf{z} \gets min(Z)$ \textit{Assign to $\mathbf{z}$ the point in $Z$ closest to $\mathbf{q}$}
        \State $\mathbf{t} \gets max(\mathcal{Q})$ \textit{Assign to $\mathbf{t}$ the point furthest from $\mathbf{q}$ in $\mathcal{Q}$}
    \EndWhile
\end{algorithmic}
\label{alg:query_update}
\end{algorithm}

The cost of updating $\mathbf{s}$ in this way is proportional to the number of points added and removed from $\mathcal{Q}$. Therefore, Algorithm \ref{alg:query_update} only saves time if there are fewer than $k$ insertions and removals to $\mathcal{Q}$ at any one time. If more than $k$ operations are required, then computing the centroid in the naive fashion becomes the faster option. In Appendix \ref{sec:update_frequency}, we show experimentally that there are often far fewer than $k$ insertions and removals to $\mathcal{Q}$ after an RP-Tree is searched, thereby validating the practical applicability of this technique.

\section{Experiments}
\label{sec:experiments}

We experimentally verify claims made in the previous sections by performing the following experiments:

\begin{enumerate}
    \item In Section \ref{sec:results}, we demonstrate the effect of modifying RP-Forest using the techniques outlined in Section \ref{sec:application_to_forest}. We show the efficacy of our technique by measuring decreased build and query time compared to the baseline algorithm.
    \item In Section \ref{sec:q_c_dissimilar_experiment}, we experimentally verify that $\mathbf{q}$ and $\normc$ are often dissimilar.
    \item In Section \ref{sec:eliminating_hash_failures}, we provide evidence that hash failures are essentially eliminated when $\normc$ is used as input to hash functions.
    \item In Section \ref{sec:estimating_c}, we measure the degree of approximation the set $\widehat{knn(\mathbf{q})}$ can possess before a useful instance of $\normchat$ can no longer be derived, and demonstrate that only a coarse estimation of the set $\knnq$ is required to compute a useful instance of $\normchat$.
\end{enumerate}

\subsection{Datasets Used}
\label{sec:datasets}

This paper studies four large, high-dimensional datasets of differing modalities, which we describe in detail in Appendix \ref{sec:datasets_appendix}.

\subsection{Experimental Comparison of MQ-Forest and RP-Forest}
\label{sec:results}

\begin{figure*}[t]
    \centering

    \begin{subfigure}{0.495\textwidth}
        \centering
        \includegraphics[width=\linewidth]{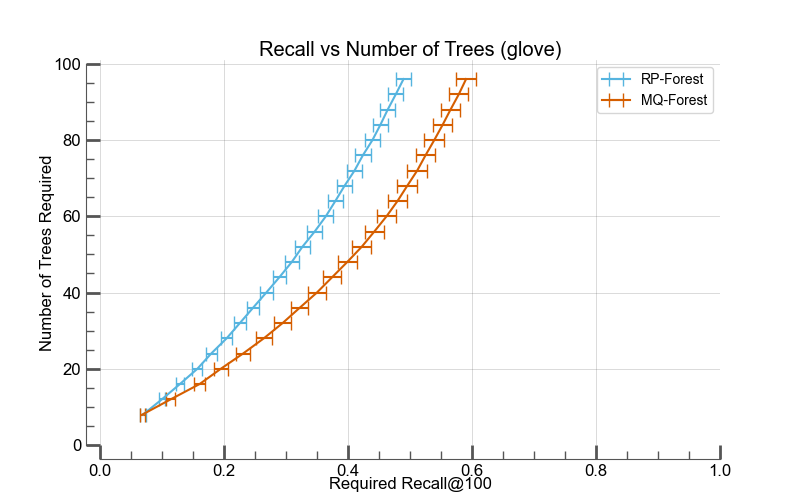}
        \caption{GloVe Recall}
    \end{subfigure}
    \hfill
    \begin{subfigure}{0.495\textwidth}
        \centering
        \includegraphics[width=\linewidth]{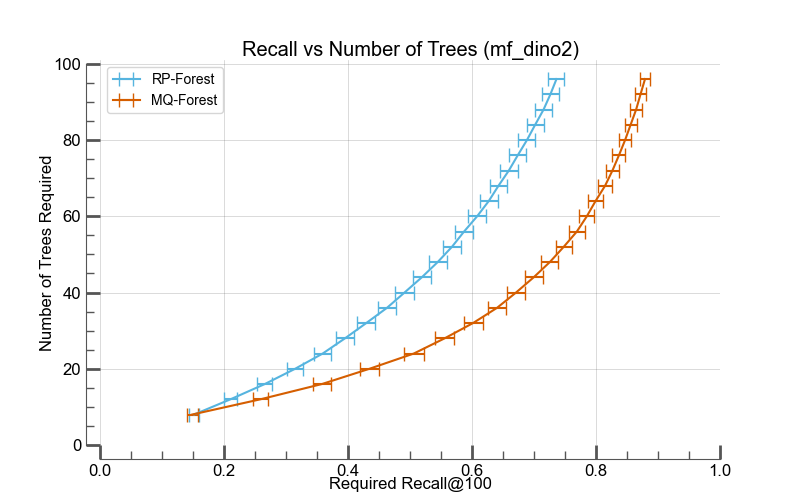}
        \caption{MirFlickr Dino2 Recall}
    \end{subfigure}

    \vspace{0.8em}

    \begin{subfigure}{0.495\textwidth}
        \centering
        \includegraphics[width=\linewidth]{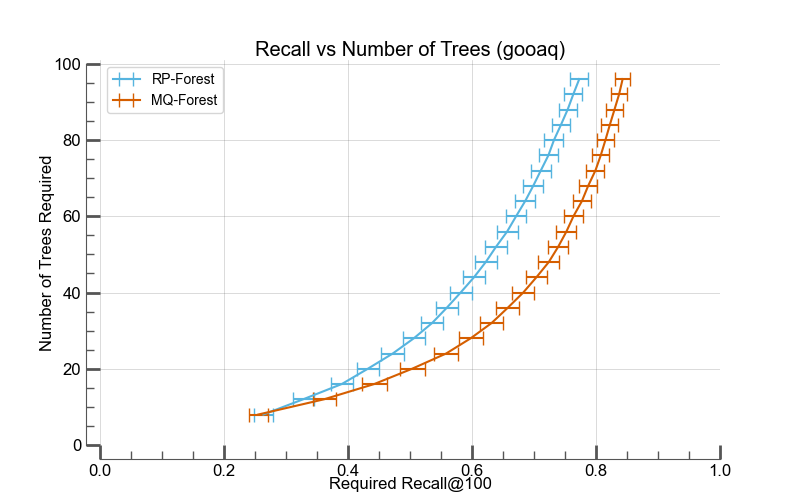}
        \caption{GooAQ Recall}
    \end{subfigure}
    \hfill
    \begin{subfigure}{0.495\textwidth}
        \centering
        \includegraphics[width=\linewidth]{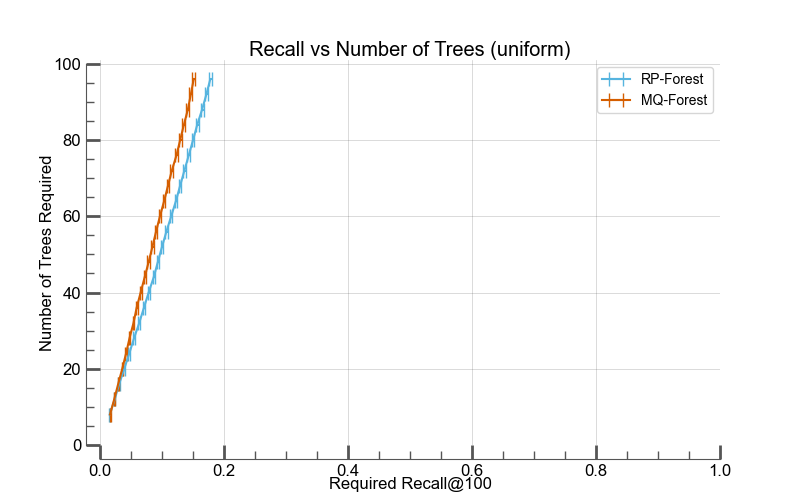}
        \caption{Uniform Recall}
    \end{subfigure}

    \caption{Recall results for an RP-Forest and MQ-Forest consisting of $\mathbf{y}$ trees. Note that the degree of vertical difference between the two curves represents improvement; lower is better. For example, MirFlickr can be searched at a recall of approximately $75\%$ in 96 trees for RP-Forest, and approximately $43$ for MQ-Forest whilst reducing the number of distance calculations.}
    \label{fig:mq_forest_recall_results}
\end{figure*}

\begin{figure*}[t]
    \centering

    \begin{subfigure}{0.495\textwidth}
        \centering
        \includegraphics[width=\linewidth]{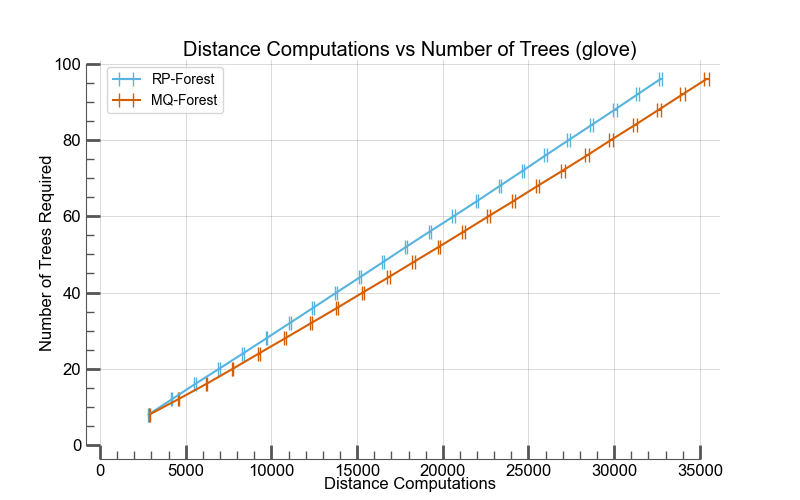}
        \caption{GloVe Distance Computations}
    \end{subfigure}
    \hfill
    \begin{subfigure}{0.495\textwidth}
        \centering
        \includegraphics[width=\linewidth]{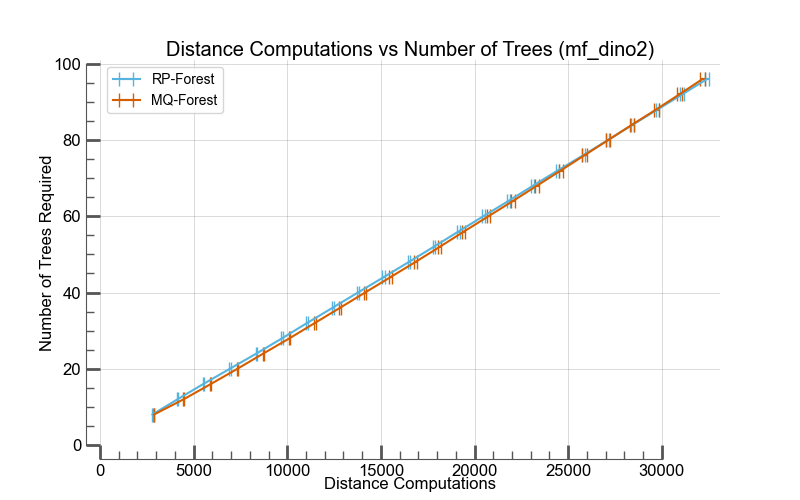}
        \caption{MirFlickr Dino2 Distance Computations}
    \end{subfigure}

    \vspace{0.8em}

    \begin{subfigure}{0.495\textwidth}
        \centering
        \includegraphics[width=\linewidth]{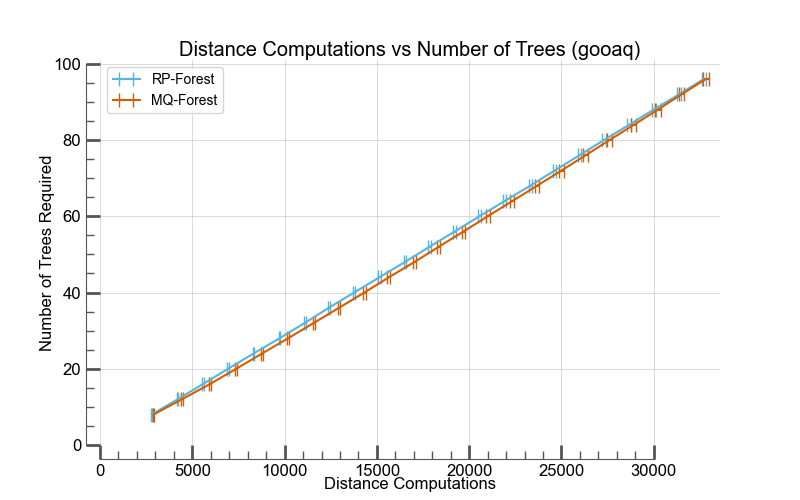}
        \caption{GooAQ Distance Computations}
    \end{subfigure}
    \hfill
    \begin{subfigure}{0.495\textwidth}
        \centering
        \includegraphics[width=\linewidth]{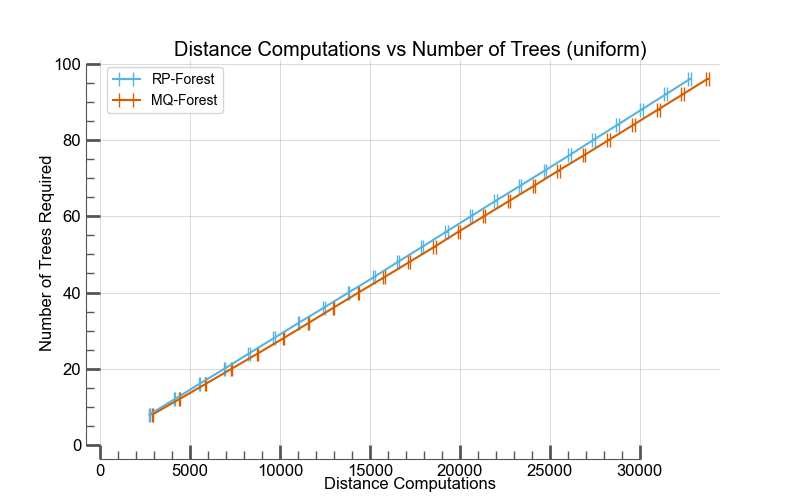}
        \caption{Uniform Distance Computations}
    \end{subfigure}

    \caption{Distance computations for an RP-Forest and MQ-Forest consisting of $\mathbf{y}$ trees.}
    \label{fig:mq_forest_comps_results}
\end{figure*}

We measure the performance of MQ-Forest, and compare it to the baseline RP-Forest algorithm. The MQ-Forest algorithm operates as follows: the first $v$ trees of the RP and MQ-Forests execute identically (where $v$ is a user-defined parameter). Upon reaching the $v+1$th tree, the MQ-Forest uses the procedure outlined in Figure \ref{fig:disentangled_lsh}, where the initial instance of $\normchat$ is generated from the candidate set formed by searching the first $v$ trees. A new estimate of $\normchat$ is computed after each tree is searched, using the fast centroid computation algorithm as described in Algorithm \ref{alg:query_update}. To measure the performance of each algorithm, we formulate the following experiment:

\begin{enumerate}
    \item Draw and exclude $250$ objects from each dataset for use as queries.
    \item Build an RP-Forest and an MQ-Forest over the dataset, with a number of trees ranging from $8$ to $96$. We fix $n_s = 500$ and $v=8$. To make a direct comparison, we use exactly the same trees and associated hash functions for both the RP and MQ-Forest. 
    \item We use each object reserved as a query as input to both the RP and MQ-Forest. For each query, we record the recall for $k=100$, and the number of distance computations required to serve the query. Note that we consider each vector addition as required by Algorithm \ref{alg:query_update} as a distance computation to ensure that the cost of centroid computation is taken into account.
\end{enumerate}

Figure \ref{fig:mq_forest_recall_results} shows the number of trees which must be built and searched in order to achieve a fixed recall value. Figure \ref{fig:mq_forest_comps_results} shows the number of distance computations required to search some number of trees.

\subsubsection{Discussion}
We find that, for all neural-network-derived datasets, the MQ-Forest achieves a higher average recall whilst requiring approximately the same number of distance computations to serve a query as RP-Forest with a fixed number of trees.

From inspection of Figure \ref{fig:mq_forest_recall_results}, it can be seen that RP-Forest must build and search up to $78\%$ more trees to achieve a fixed recall result than MQ-Forest, depending on the dataset. Thus, MQ-Forest can expect expedited build and query times comparatively.

For example, for the GloVe dataset, RP-Forest achieves a recall of $0.49$ after $96$ trees have been searched. MQ-Forest requires only $68$ trees be built and searched to achieve the same recall, a reduction of $29\%$. The number of metric calls required for this search are also reduced; 32696 are required for the RP-Forest, compared with 25543 for the MQ-Forest, a $22\%$ decrease. Our choice to use the GloVe dataset as an example here is deliberate, as it benefits the least of all neural-network-derived benchmark datasets from the query modification framework. For all other datasets, the MQ-Forest performs better relatively.

However, for the uniform dataset, the MQ-Forest fails to improve performance. An explanation for poor performance of MQ-Forest in this scenario can be seen in Figure \ref{fig:m_near_neighbour_goodness}. It shows that it is difficult in the context of the uniform dataset to form a good instance of $\normchat$ without a very good instance of $\widehat{knn(\mathbf{q})}$. As RP-Forest does not return a significant number of near neighbours for the uniform dataset, an insufficiently good estimate of $\widehat{knn(\mathbf{q})}$ could not be formed, thus leading MQ-Forest to fail.

\subsubsection{Fast Production Implementation}

We provide a production-ready Cython library for Python, which is available for download from \url{https://github.com/benclaydon/rpforest_mtq/}. This version of MQ-Forest is derived from the RP-Forest implementation available from \url{https://github.com/lyst/rpforest}. In Figure \ref{fig:production_speeds}, we compare the query throughput of the existing RP-Forest system with our modification for each embedding dataset. We show that our system achieves a higher query throughput for each embedding dataset. We once again emphasise that, as shown in Section \ref{sec:results}, the MQ-Forest requires significantly less build time than RP-Forest.

\begin{figure}
    \centering
    \begin{subfigure}{0.48\textwidth}
        \centering
        \includegraphics[width=\linewidth]{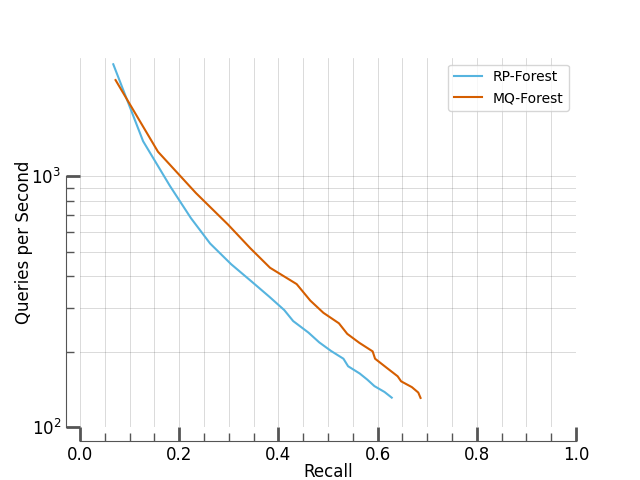}
        \caption{GloVe}
        \label{fig:production_glove}
    \end{subfigure}
    \hfill
    \begin{subfigure}{0.48\textwidth}
        \centering
        \includegraphics[width=\linewidth]{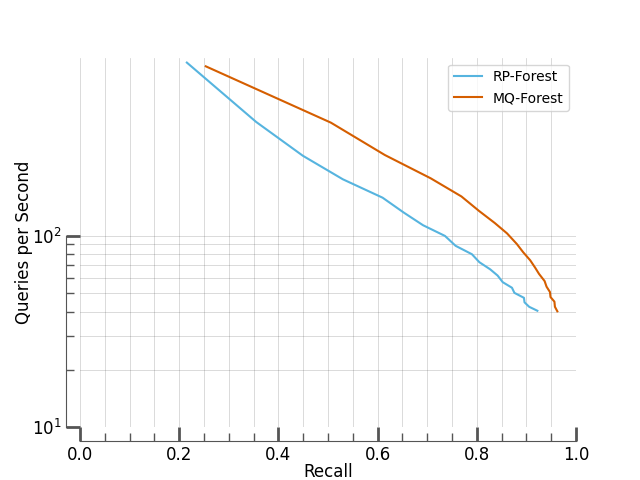}
        \caption{DINOv2}
        \label{fig:production_dino2}
    \end{subfigure}

    \vspace{1em}
    \begin{subfigure}{0.48\textwidth}
        \centering
        \includegraphics[width=\linewidth]{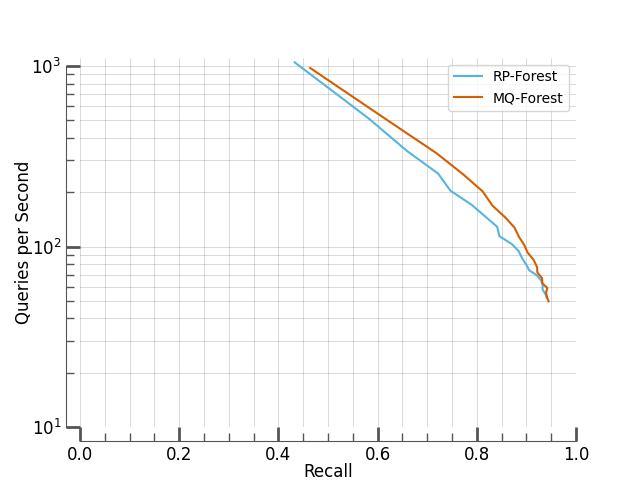}
        \caption{GooAQ}
        \label{fig:production_gooaq}
    \end{subfigure}
    \hfill
    \begin{subfigure}{0.48\textwidth}
        \centering
            \includegraphics[width=\linewidth]{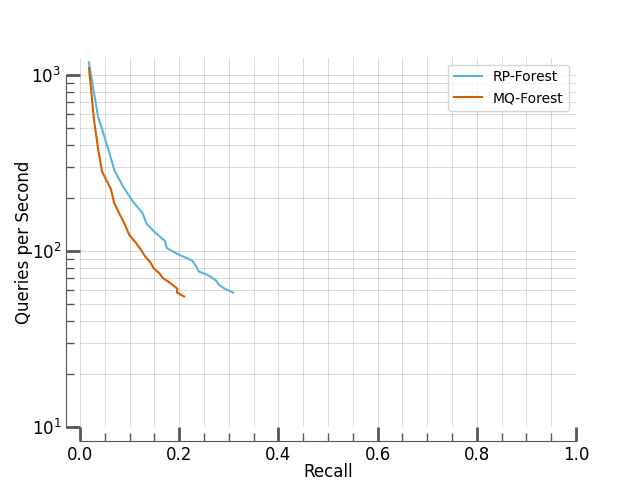}
        \caption{Uniform 200D}
        \label{fig:production_uniform}
    \end{subfigure}

    \caption{RP-Forest and MQ-Forest: Recall vs Queries Per Second. Note the log-scale $y$ axis. Searching for $k=20$ near neighbours.}
    \label{fig:production_speeds}
\end{figure}

\subsection{Observations on Distance Between $\mathbf{q}$ and $\normc$}
\label{sec:q_c_dissimilar_experiment}

We show that $\mathbf{q}$ and $\normc$ are often distant by computing their distance over a sample of 100 randomly drawn queries from each embedding dataset. The scale of these distances is large enough to cause $\mathbf{q}$ and $\normc$ to frequently take different hash values. As an example, consider the smallest average distance of any dataset, which is $0.33$ for the uniform dataset. Using Charikar's distance to hash collision probability formula \cite{Charikar_2002}, we may compute that $\Pr(h(\mathbf{q}) = h(\mathbf{\normc})) = 0.89$ at a distance of $0.33$. For some $m$ bit compound hash function, the probability of collision is $0.89^m$, which quickly approaches $0$ as $m$ grows. Therefore, the set of points retrieved by hashing $\mathbf{q}$ and $\normc$ becomes increasingly likely to differ. However, we have shown that those points retrieved by $\normc$ are expected to be superior.

To aid comparison, we also show in Table \ref{tab:distances} the mean distances from $\mathbf{q}$ to its 1nn, 10nn, 100nn. We also record the mean distance to all other points in the dataset. Mean values of these distances are shown in Table \ref{tab:distances}.

\begin{table}[t]
\centering
\begin{tabular}{l|c|c|c|c}
\hline
 & Glove & MirFlickr & GOOAQ & Uniform 200D \\
\hline
$||\mathbf{q} - \left< \mathbf{c} \right>||_2$ 
& 0.59 & 0.58 & 0.44 & 0.33 \\

$||\mathbf{q} - \mathbf{x}_1||_2$ 
& 0.77 & 0.71 & 0.31 & 1.15 \\

$||\mathbf{q} - \mathbf{x}_{10}||_2$ 
& 0.84 & 0.78 & 0.49 & 1.18 \\

$||\mathbf{q} - \mathbf{x}_{100}||_2$ 
& 0.91 & 0.86 & 0.69 & 1.20 \\

\hline

Average Interpoint Distance 
& 1.31 & 1.39 & 1.38 & 1.41 \\
\hline

\end{tabular}

\caption{
Mean Euclidean distance from $\mathbf{q}$ to $\normc$ and its $k$-th nearest neighbour for each embedding dataset and $k=1, 10, 100$. These distances are average observations measured using $100$ randomly sampled vectors from each dataset.
}

\label{tab:distances}
\end{table}

\subsection{Eliminating Hash Failure Cases}
\label{sec:eliminating_hash_failures}
Here we demonstrate experimentally that hash failure cases are essentially eliminated when $\normc$ is used as input to a binary hash function. To show this we formulate an experiment by initially randomly drawing and excluding a set of $100$ queries $Q \subset X$ from each dataset, and a subset $H \subset \mathcal{H}_{rp}$ of $1000$ random binary hash functions are drawn.

For all combinations $(h \in H) \times (\mathbf{q} \in Q)$, we measure what fraction of $100nn(\mathbf{q})$ collide with $\mathbf{q}$ and $\normc$. This results in $100,000$ individual recall measurements for each dataset. A hash failure is recorded to have occurred at recall values close to $0$. We report the distributions of these recall values in Figure \ref{fig:hash_collisions} for each embedding dataset.

\begin{figure}
    \centering
    
    \begin{subfigure}{0.38\textwidth}
        \centering
        \includegraphics[width=\linewidth]{fig/hash_probs/mf_dino2_q.png}
        \caption{Dino2 ($\mathbf{q}$)}
        \label{fig:dino2_q_example}
    \end{subfigure}
    \hfill
    \begin{subfigure}{0.38\textwidth}
        \centering
        \includegraphics[width=\linewidth]{fig/hash_probs/mf_dino2_c.png}
        \caption{Dino2 ($\normc$)}
        \label{fig:dino2_c_example}
    \end{subfigure}

    \vspace{0.5em}

    \begin{subfigure}{0.38\textwidth}
        \centering
        \includegraphics[width=\linewidth]{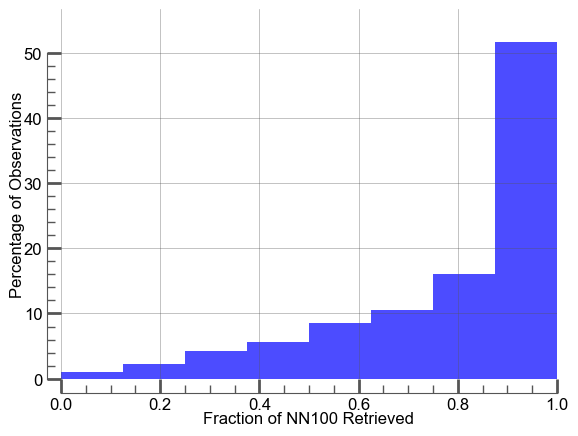}
        \caption{GooAQ ($\mathbf{q}$)}
        \label{fig:gooaq_q_example}
    \end{subfigure}
    \hfill
    \begin{subfigure}{0.38\textwidth}
        \centering
        \includegraphics[width=\linewidth]{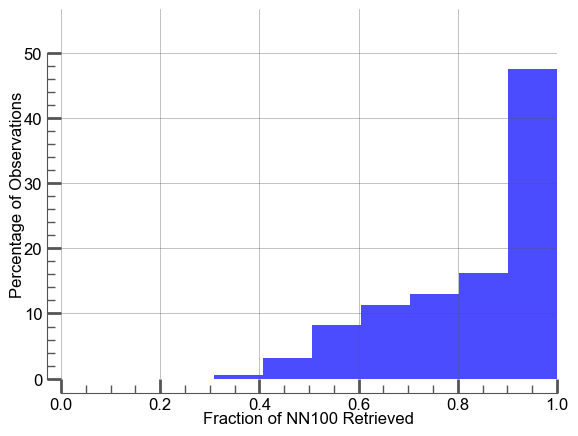}
        \caption{GooAQ ($\normc$)}
        \label{fig:gooaq_c_example}
    \end{subfigure}

    \vspace{0.5em}

    \begin{subfigure}{0.38\textwidth}
        \centering
        \includegraphics[width=\linewidth]{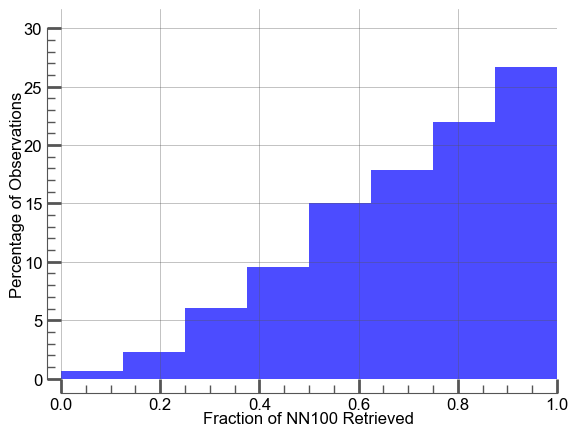}
        \caption{GloVe ($\mathbf{q}$)}
        \label{fig:glove_q_example}
    \end{subfigure}
    \hfill
    \begin{subfigure}{0.38\textwidth}
        \centering
        \includegraphics[width=\linewidth]{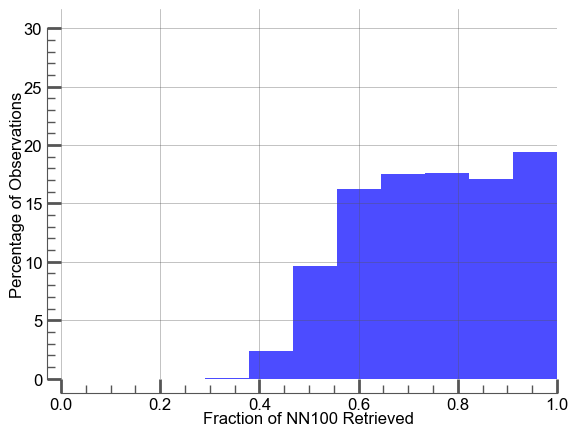}
        \caption{GloVe ($\normc$)}
        \label{fig:glove_c_example}
    \end{subfigure}

    \begin{subfigure}{0.38\textwidth}
        \centering
        \includegraphics[width=\linewidth]{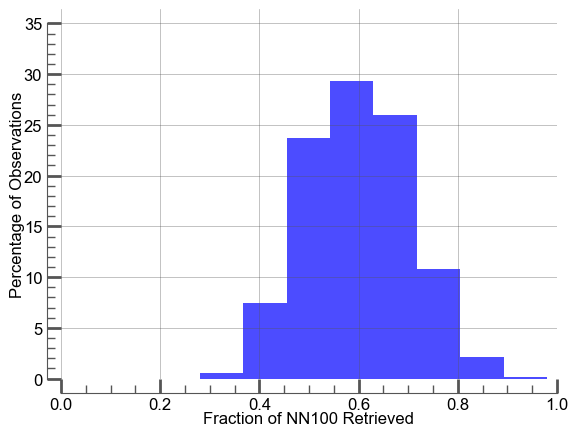}
        \caption{Uniform ($\mathbf{q}$)}
        \label{fig:uniform_q_example}
    \end{subfigure}
    \hfill
    \begin{subfigure}{0.38\textwidth}
        \centering
        \includegraphics[width=\linewidth]{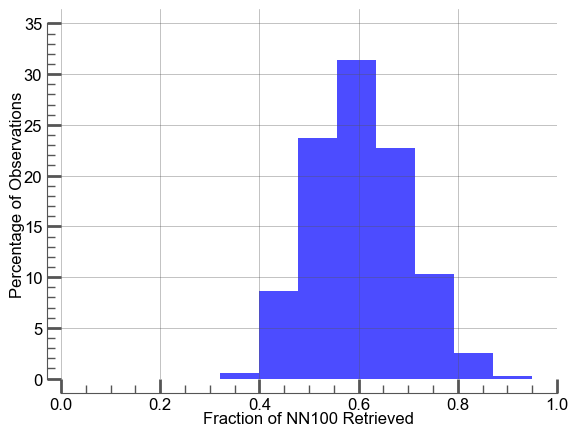}
        \caption{Uniform ($\normc$)}
        \label{fig:uniform_c_example}
    \end{subfigure}

    \caption{Fraction of points in $\knnq$ which collide with $\mathbf{q}$ (left column) and $\normc$ (right column). Measured using $1000$ randomly drawn binary hash functions and $100$ queries.}
    \label{fig:hash_collisions}
\end{figure}


We observe that there are no instances of hash failures when $\normc$ is used as input to randomly drawn binary hash functions for any embedding dataset. This experimentally reinforces the result set out in Section \ref{sec:off_diag_properties}, where the correlation of hash function outputs was explained to be the cause of hash failures. In Appendix \ref{sec:off-diag}, we further demonstrate that the average value of $\operatorname{Cov}
\left(
    \mathbf{x}\cdot\mathbf{w},
    \mathbf{y}\cdot\mathbf{w} | \mathbf{w} \cdot \mathbf{u} = b
\right)$ for some $\mathbf{x}, \mathbf{y} \in \knnq$, some hyperplane normal $\mathbf{w}$, and some constant $b$ is approximately $0$ when $\mathbf{u} = \normchat$. 

\subsection{Varying Degrees of Approximations of $\widehat{knn(\mathbf{q})}$}
\label{sec:estimating_c}

Here we study how well instances of $\normchat$ perform when derived from sets $\widehat{knn(\mathbf{q})}$ of various degrees of approximation. If a good instance of $\normchat$ can be drawn from a weak approximation of $\widehat{knn(\mathbf{q})}$, then it follows that we may derive a good $\normchat$ from the result of a coarse search using only a small number of compound hash functions. Initially, we define a quality metric $\kappa$ to measure the quality of a given $\normchat$:

\[
\kappa \eqdef \frac{||\left< \mathbf{c} \right> - \normchat||_2}{||\left< \mathbf{c} \right> - \mathbf{q}||_2}
\]

Under this metric, if $\normchat = \left< \mathbf{c} \right>$, $\kappa = 0$; this represents the best possible $\normchat$. If $\normchat = \mathbf{q}$, then $\kappa = 1$, which represents that $\normchat$ is expected to have equal performance to the original query. If $0 < \kappa < 1$, it is probable that the instance of $\normchat$ is a member of $\Phi$ and thus will perform favourably to $\mathbf{q}$, but not so well as $\left< \mathbf{c} \right>$. Any $\normchat$ with $\kappa > 1$ is worse than the query, and is therefore of no use. Any $\normchat$ with a $\kappa$ value of less than $1$ is likely to be a member of the set $\Phi$ as it is closer to $\normc$ than $\mathbf{q}$ is to $\normc$.

We record $\kappa$ values of $\normchat$ instances which were derived from sets containing $k=100$ elements randomly sampled from the $m$-near neighbour set. We call such a set $Y_m$. By increasing $m$, the set $Y_m$ will contain elements less similar to $\mathbf{q}$ on average.

We form an experiment by randomly drawing and excluding $250$ elements from each dataset.  Sets $Y_m$ with values of $m$ between $100$ and $5000$ were then created for each query (apart from queries drawn from the uniform 200D dataset, as the value of $\kappa$ grows quickly with $m$ for this dataset; therefore we cap $m=500$ for this dataset). Finally, $\kappa$ values were measured from each $Y_m$. Figure \ref{fig:m_near_neighbour_goodness} shows the average $\kappa$ values compared with the value of $m$.

\begin{figure}
    \centering
    \begin{subfigure}[b]{0.45\linewidth}
        \centering
        \includegraphics[width=\linewidth]{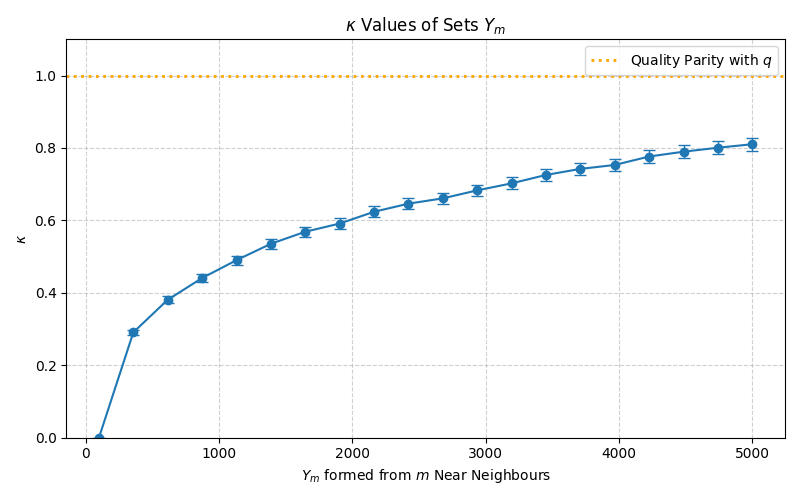}
        \caption{MirFlickr DINOv2}
        \label{fig:dino2_faroutness}
    \end{subfigure}
    
    \begin{subfigure}[b]{0.45\linewidth}
        \centering
        \includegraphics[width=\linewidth]{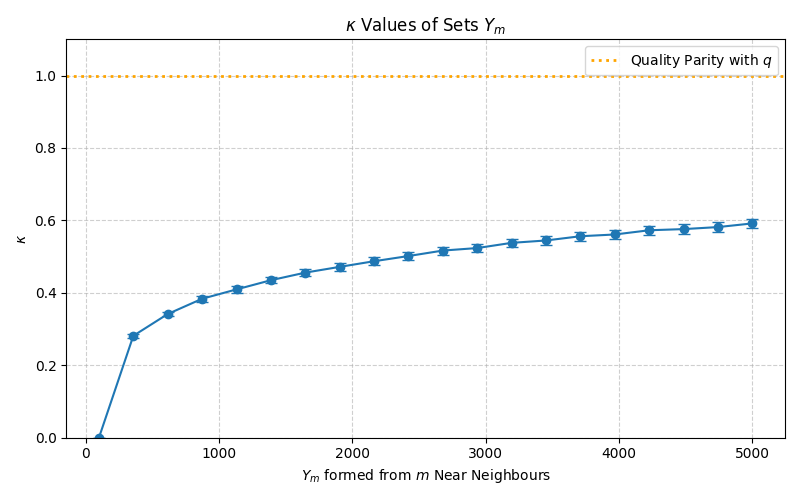}
        \caption{Twitter GloVe}
        \label{fig:glove_faroutness}
    \end{subfigure}

        \begin{subfigure}[b]{0.45\linewidth}
        \centering
        \includegraphics[width=\linewidth]{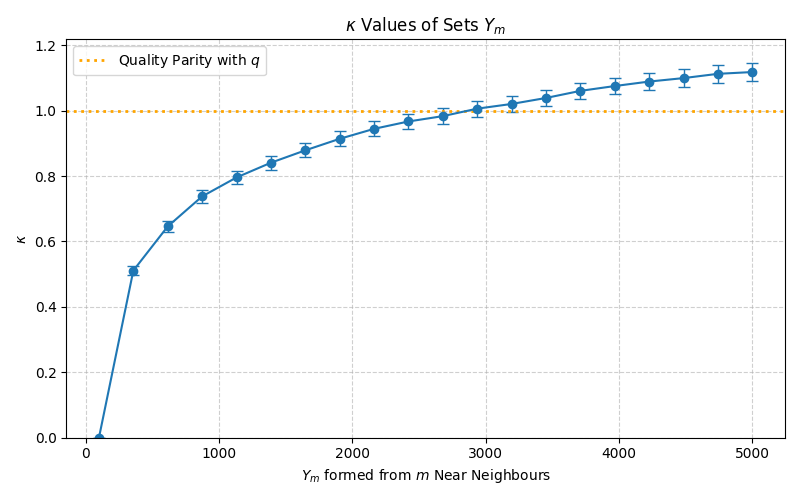}
        \caption{GOOAQ}
        \label{fig:gooaq_faroutness}
    \end{subfigure}

    \begin{subfigure}[b]{0.45\linewidth}
        \centering
        \includegraphics[width=\linewidth]{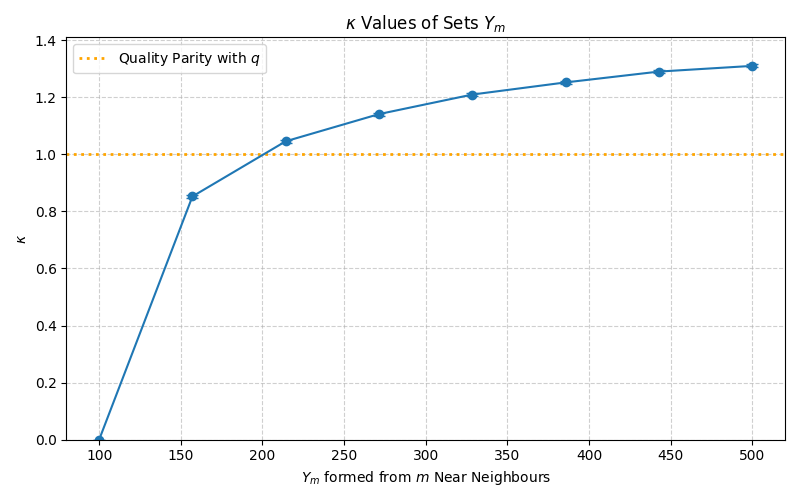}
        \caption{Uniform 200D}
        \label{fig:uniform_faroutness}
    \end{subfigure}
    
    \caption{Average value of $\kappa$ from estimating $\normc$ via a random selection of $100$ elements from the $m$ nearest neighbours.}
    \label{fig:m_near_neighbour_goodness}
\end{figure}

We show that, even when $m$ is significantly larger than $k$, a useful instance of $\normchat$ may be derived from the associated set $Y_m$. An exception to this is the uniform 200D dataset, where the value of $m$ can be at most twice $k$ if a good estimate of $\normchat$ is to be derived. This observation reinforces the findings of Appendix \ref{sec:clt_proof}, which theoretically show why taking the normalised centroid of some $\widehat{\knnq}$ is likely to produce some $\normchat \in \Phi$.

Additionally, we show that the average value of $\kappa$ decreases monotonically as the quality of $Y_m$ improves. In the context of an LSH-based mechanism, the candidate set similarly monotonically improves as new candidates are retrieved and re-ranked. Thus, we can expect the estimate $\normchat$ in the MQ-Forest to similarly monotonically improve as it is iteratively refined. 


\section{Conclusion}
In this paper, we have shown that the build and query procedure of LSH-based ANN mechanisms may be improved via iterative query modification whilst leaving the build procedure entirely unchanged. We show that iterative query modification enjoys a strong statistical and geometric underpinning, as well as positive practical outcomes when applied to a suite of large, high-dimensional benchmark datasets.

\section{Resources}

All code used for generating results for this paper can be found at \url{https://github.com/benclaydon/rp-forest-moving-query/tree/unrestricted-euc}.

\bibliographystyle{plain}
\bibliography{references}

@inproceedings{rptree, address={Victoria British Columbia Canada}, title={Random projection trees and low dimensional manifolds}, ISBN={978-1-60558-047-0}, url={https://dl.acm.org/doi/10.1145/1374376.1374452}, DOI={10.1145/1374376.1374452}, abstractNote={We present a simple variant of the k-d tree which automatically adapts to intrinsic low dimensional structure in data without having to explicitly learn this structure.}, booktitle={Proceedings of the fortieth annual ACM symposium on Theory of computing}, publisher={ACM}, author={Dasgupta, Sanjoy and Freund, Yoav}, year={2008}, month=may, pages={537–546}, language={en} }

@inproceedings{Yan_Wang_Wang_Wang_Li_2018, title={K-nearest Neighbor Search by Random Projection Forests}, url={http://arxiv.org/abs/1812.11689}, DOI={10.1109/BigData.2018.8622307}, abstractNote={K-nearest neighbor (kNN) search has wide applications in many areas, including data mining, machine learning, statistics and many applied domains. Inspired by the success of ensemble methods and the flexibility of tree-based methodology, we propose random projection forests (rpForests), for kNN search. rpForests finds kNNs by aggregating results from an ensemble of random projection trees with each constructed recursively through a series of carefully chosen random projections. rpForests achieves a remarkable accuracy in terms of fast decay in the missing rate of kNNs and that of discrepancy in the kNN distances. rpForests has a very low computational complexity. The ensemble nature of rpForests makes it easily run in parallel on multicore or clustered computers; the running time is expected to be nearly inversely proportional to the number of cores or machines. We give theoretical insights by showing the exponential decay of the probability that neighboring points would be separated by ensemble random projection trees when the ensemble size increases. Our theory can be used to refine the choice of random projections in the growth of trees, and experiments show that the effect is remarkable.}, note={arXiv:1812.11689 [cs, stat]}, booktitle={2018 IEEE International Conference on Big Data (Big Data)}, author={Yan, Donghui and Wang, Yingjie and Wang, Jin and Wang, Honggang and Li, Zhenpeng}, year={2018}, month=dec, pages={4775–4781} }

@inproceedings{Charikar_2002, series={STOC ’02}, title={Similarity estimation techniques from rounding algorithms}, ISBN={978-1-58113-495-7}, abstractNote={(MATH) A locality sensitive hashing scheme is a distribution on a family $F$ of hash functions operating on a collection of objects, such that for two objects x,y, PrhεF[h(x) = h(y)] = sim(x,y), where sim(x,y) ε [0,1] is some similarity function defined on the collection of objects. Such a scheme leads to a compact representation of objects so that similarity of objects can be estimated from their compact sketches, and also leads to efficient algorithms for approximate nearest neighbor search and clustering. Min-wise independent permutations provide an elegant construction of such a locality sensitive hashing scheme for a collection of subsets with the set similarity measure sim(A,B) = frac{|A ∩ B|}{|A ∪ B|}.(MATH) We show that rounding algorithms for LPs and SDPs used in the context of approximation algorithms can be viewed as locality sensitive hashing schemes for several interesting collections of objects. Based on this insight, we construct new locality sensitive hashing schemes for:A collection of vectors with the distance between → over u and → over v measured by Ø(→ over u, → over v)/π, where Ø(→ over u, → over v) is the angle between → over u) and → over v). This yields a sketching scheme for estimating the cosine similarity measure between two vectors, as well as a simple alternative to minwise independent permutations for estimating set similarity.A collection of distributions on n points in a metric space, with distance between distributions measured by the Earth Mover Distance (EMD), (a popular distance measure in graphics and vision). Our hash functions map distributions to points in the metric space such that, for distributions P and Q, EMD(P,Q) ≤ EhεF [d(h(P),h(Q))] ≤ O(log n log log n). EMD(P, Q).}, booktitle={Proceedings of the thiry-fourth annual ACM symposium on Theory of computing}, publisher={Association for Computing Machinery}, author={Charikar, Moses S.}, pages={380–388}, collection={STOC ’02} }

@inproceedings{Gondara_2015, title={Random Forest with Random Projection to Impute Missing Gene Expression Data}, url={https://ieeexplore.ieee.org/abstract/document/7424493}, DOI={10.1109/ICMLA.2015.29}, abstractNote={Measurement error or lack of proper experimental setup often results in invalid or missing data in gene expression studies. Small sample size and cost of re-running the experiment presents a need for an efficient missing data imputation technique. In this paper, we propose a method based on Random forest using Random projection as a data pre-processing filter. Initial results using varying missing data proportions on variety of real datasets show that the imputation process based on Random forest performs equally well or better than K-Nearest Neighbor & Support Vector Regression based methods. Using Random projection we show that dimensionality of a dataset can be reduced by 50 percent without affecting the imputation process.}, booktitle={2015 IEEE 14th International Conference on Machine Learning and Applications (ICMLA)}, author={Gondara, Lovedeep}, year={2015}, month=dec, pages={1251–1256} }

@inproceedings{Lee_Yang_Oh_2015, address={Santiago, Chile}, title={Fast and Accurate Head Pose Estimation via Random Projection Forests}, ISBN={978-1-4673-8391-2}, url={http://ieeexplore.ieee.org/document/7410584/}, DOI={10.1109/ICCV.2015.227}, abstractNote={In this paper, we consider the problem of estimating the gaze direction of a person from a low-resolution image. Under this condition, reliably extracting facial features is very difﬁcult. We propose a novel head pose estimation algorithm based on compressive sensing. Head image patches are mapped to a large feature space using the proposed extensive, yet efﬁcient ﬁlter bank. The ﬁlter bank is designed to generate sparse responses of color and gradient information, which can be compressed using random projection, and classiﬁed by a random forest. Extensive experiments on challenging datasets show that the proposed algorithm performs favorably against the state-of-the-art methods on head pose estimation in low-resolution images degraded by noise, occlusion, and blurring.}, booktitle={2015 IEEE International Conference on Computer Vision (ICCV)}, publisher={IEEE}, author={Lee, Donghoon and Yang, Ming-Hsuan and Oh, Songhwai}, year={2015}, month=dec, pages={1958–1966}, language={en} }

@article{Lee_Yang_Oh_2019, title={Head and Body Orientation Estimation Using Convolutional Random Projection Forests}, volume={41}, ISSN={1939-3539}, DOI={10.1109/TPAMI.2017.2784424}, abstractNote={In this paper, we consider the problem of estimating the head pose and body orientation of a person from a low-resolution image. Under this setting, it is difficult to reliably extract facial features or detect body parts. We propose a convolutional random projection forest (CRPforest) algorithm for these tasks. A convolutional random projection network (CRPnet) is used at each node of the forest. It maps an input image to a high-dimensional feature space using a rich filter bank. The filter bank is designed to generate sparse responses so that they can be efficiently computed by compressive sensing. A sparse random projection matrix can capture most essential information contained in the filter bank without using all the filters in it. Therefore, the CRPnet is fast, e.g., it requires 0.04;mathrmms to process an image of 50times 50 pixels, due to the small number of convolutions (e.g., 0.01 percent of a layer of a neural network) at the expense of less than 2 percent accuracy. The overall forest estimates head and body pose well on benchmark datasets, e.g., over 98 percent on the HIIT dataset, while requiring 3.8;mathrmms without using a GPU. Extensive experiments on challenging datasets show that the proposed algorithm performs favorably against the state-of-the-art methods in low-resolution images with noise, occlusion, and motion blur.}, number={1}, journal={IEEE Transactions on Pattern Analysis and Machine Intelligence}, author={Lee, Donghoon and Yang, Ming-Hsuan and Oh, Songhwai}, year={2019}, month=jan, pages={107–120} }

@article{Tan_Yang_Rahardja_2022, title={Sparse random projection isolation forest for outlier detection}, volume={163}, ISSN={0167-8655}, DOI={10.1016/j.patrec.2022.09.015}, abstractNote={Isolation Forest has a low computational complexity, hence has been widely applied to detect outliers in large-scale data. However, it suffers from the artifacts caused by the hyperplanes chosen, thereby failing to detect outliers in some specific regions. To tackle this problem, we propose the random-projection-based Isolation Forest, which works in two steps. First, we transform the data using the random projection technique. Then, we employ the Isolation Forest to identify outliers using the transformed data. Experimental results show that the proposed methods outperform 12 state-of-the-art outlier detectors.}, journal={Pattern Recognition Letters}, author={Tan, Xu and Yang, Jiawei and Rahardja, Susanto}, year={2022}, month=nov, pages={65–73} }

@article{Zhang_2011, title={Phenotype Recognition for RNAi Screening by Random Projection Forest}, volume={1371}, ISSN={0094-243X}, DOI={10.1063/1.3596627}, abstractNote={High‐content screening is important in drug discovery. The use of images of living cells as the basic unit for molecule discovery can aid the identification of small compounds altering cellular phenotypes. As such, efficient computational methods are required for the rate limiting task of cellular phenotype identification. In this paper we first investigate the effectiveness of a feature description approach by combining Haralick texture analysis with Curvelet transform and then propose a new ensemble approach for classification. The ensemble contains a set of base classifiers which are trained using random projection (RP) of original features onto higher‐dimensional spaces. With Classification and Regression Tree (CART) as the base classifier, it has been empirically demonstrated that the proposed Random Projection Forest ensemble gives better classification results than those achieved by the Boosting, Bagging and Rotation Forest algorithms, offering a classification rate ∼88 percent with smallest standard deviation, which compares sharply with the published result of 82 percent.}, number={1}, journal={AIP Conference Proceedings}, author={Zhang, Bailing}, year={2011}, month=june, pages={55–64} }

@inproceedings{Faggioli_Ferro_Perego_Tonellotto_2024, address={New York, NY, USA}, series={SIGIR ’24}, title={Dimension Importance Estimation for Dense Information Retrieval}, ISBN={979-8-4007-0431-4}, url={https://dl.acm.org/doi/10.1145/3626772.3657691}, DOI={10.1145/3626772.3657691}, abstractNote={Recent advances in Information Retrieval have shown the effectiveness of embedding queries and documents in a latent high-dimensional space to compute their similarity. While operating on such high-dimensional spaces is effective, in this paper, we hypothesize that we can improve the retrieval performance by adequately moving to a query-dependent subspace. More in detail, we formulate the Manifold Clustering (MC) Hypothesis: projecting queries and documents onto a subspace of the original representation space can improve retrieval effectiveness. To empirically validate our hypothesis, we define a novel class of Dimension IMportance Estimators (DIME). Such models aim to determine how much each dimension of a high-dimensional representation contributes to the quality of the final ranking and provide an empirical method to select a subset of dimensions where to project the query and the documents. To support our hypothesis, we propose an oracle DIME, capable of effectively selecting dimensions and almost doubling the retrieval performance. To show the practical applicability of our approach, we then propose a set of DIMEs that do not require any oracular piece of information to estimate the importance of dimensions. These estimators allow us to carry out a dimensionality selection that enables performance improvements of up to +11.5 percent (moving from 0.675 to 0.752 nDCG@10) compared to the baseline methods using all dimensions. Finally, we show that, with simple and realistic active feedback, such as the user’s interaction with a single relevant document, we can design a highly effective DIME, allowing us to outperform the baseline by up to +0.224 nDCG@10 points (+58.6 percent, moving from 0.384 to 0.608).}, booktitle={Proceedings of the 47th International ACM SIGIR Conference on Research and Development in Information Retrieval}, publisher={Association for Computing Machinery}, author={Faggioli, Guglielmo and Ferro, Nicola and Perego, Raffaele and Tonellotto, Nicola}, year={2024}, month=july, pages={1318–1328}, collection={SIGIR ’24} }

@article{Cui_Li_Zhu_Li_Zhang_2024, title={Online Query Expansion Hashing for Efficient Image Retrieval}, volume={34}, ISSN={1558-2205}, DOI={10.1109/TCSVT.2023.3296412}, abstractNote={Unsupervised hashing has the desirable advantages of label independence, high storage, and retrieval efficiency, which is suitable for scalable image retrieval. Most existing methods focus on enhancing the image hashing model training process at the offline stage. However, little attention has been paid to the query content analysis by them at the online retrieval stage. They still suffer from important query semantic shortages, and thus limit the online retrieval performance, which is the ultimate objective of the image retrieval system. In this paper, we propose an Online Query Expansion Hashing (OQEH) for efficient image retrieval, by adaptively enhancing the discriminative capability of query hash codes in an expansion manner at the online retrieval stage. Specifically, we first design a self-expansion network to learn semantically invariant feature representations from images and their visual augmentations. Then, we conduct neighborhood-expansion to search similar samples for each image from a query expansion set with the semantically invariant features and design a Transformer architecture to adaptively transfer the semantics of neighbor samples to their corresponding images. With the support of semantically invariant features, query expansion set, and adaptive semantic transfer, the representation capability of query hash codes can be enhanced at the online retrieval stage. Experimental results demonstrate that the proposed OQEH method achieves superior retrieval accuracy and comparable retrieval efficiency compared with the state-of-the-art methods. Particularly, on MS COCO dataset, OQEH can obtain about 6 percent performance improvement compared with the state-of-the-art results. The source codes of our method are available at: https://github.com/christinecui/OQEH.}, number={3}, journal={IEEE Transactions on Circuits and Systems for Video Technology}, author={Cui, Hui and Li, Fengling and Zhu, Lei and Li, Jingjing and Zhang, Zheng}, year={2024}, month=mar, pages={1941–1953} }

@article{Huang_Feng_Zhang_Fang_Ng_2015, title={Query-aware locality-sensitive hashing for approximate nearest neighbor search}, volume={9}, ISSN={2150-8097}, DOI={10.14778/2850469.2850470}, abstractNote={Locality-Sensitive Hashing (LSH) and its variants are the well-known indexing schemes for the c-Approximate Nearest Neighbor (c-ANN) search problem in high-dimensional Euclidean space. Traditionally, LSH functions are constructed in a query-oblivious manner in the sense that buckets are partitioned before any query arrives. However, objects closer to a query may be partitioned into diﬀerent buckets, which is undesirable. Due to the use of query-oblivious bucket partition, the state-of-the-art LSH schemes for external memory, namely C2LSH and LSB-Forest, only work with approximation ratio of integer c ≥ 2.}, number={1}, journal={Proceedings of the VLDB Endowment}, author={Huang, Qiang and Feng, Jianlin and Zhang, Yikai and Fang, Qiong and Ng, Wilfred}, year={2015}, month=sept, pages={1–12}, language={en} }

@inproceedings{Kuo_Chen_Chiang_Hsu_2009, address={Beijing China}, title={Query expansion for hash-based image object retrieval}, ISBN={978-1-60558-608-3}, url={https://dl.acm.org/doi/10.1145/1631272.1631284}, DOI={10.1145/1631272.1631284}, abstractNote={An efficient indexing method is essential for content-based image retrieval with the exponential growth in large-scale videos and photos. Recently, hash-based methods (e.g., locality sensitive hashing – LSH) have been shown efficient for similarity search. We extend such hash-based methods for retrieving images represented by bags of (high-dimensional) feature points. Though promising, the hash-based image object search suffers from low recall rates. To boost the hash-based search quality, we propose two novel expansion strategies – intra-expansion and inter-expansion. The former expands more target feature points similar to those in the query and the latter mines those feature points that shall co-occur with the search targets but not present in the query. We further exploit variations for the proposed methods. Experimenting in two consumer-photo benchmarks, we will show that the proposed expansion methods are complementary to each other and can collaboratively contribute up to 76.3 percent (average) relative improvement over the original hash-based method.}, booktitle={Proceedings of the 17th ACM international conference on Multimedia}, publisher={ACM}, author={Kuo, Yin-Hsi and Chen, Kuan-Ting and Chiang, Chien-Hsing and Hsu, Winston H.}, year={2009}, month=oct, pages={65–74}, language={en} }

@article{Lv_Josephson_Wang_Charikar_Li, title={Multi-Probe LSH: Efﬁcient Indexing for High-Dimensional Similarity Search}, abstractNote={Similarity indices for high-dimensional data are very desirable for building content-based search systems for featurerich data such as audio, images, videos, and other sensor data. Recently, locality sensitive hashing (LSH) and its variations have been proposed as indexing techniques for approximate similarity search. A signiﬁcant drawback of these approaches is the requirement for a large number of hash tables in order to achieve good search quality. This paper proposes a new indexing scheme called multi-probe LSH that overcomes this drawback. Multi-probe LSH is built on the well-known LSH technique, but it intelligently probes multiple buckets that are likely to contain query results in a hash table. Our method is inspired by and improves upon recent theoretical work on entropy-based LSH designed to reduce the space requirement of the basic LSH method. We have implemented the multi-probe LSH method and evaluated the implementation with two diﬀerent high-dimensional datasets. Our evaluation shows that the multi-probe LSH method substantially improves upon previously proposed methods in both space and time eﬃciency. To achieve the same search quality, multi-probe LSH has a similar timeeﬃciency as the basic LSH method while reducing the number of hash tables by an order of magnitude. In comparison with the entropy-based LSH method, to achieve the same search quality, multi-probe LSH uses less query time and 5 to 8 times fewer number of hash tables.}, author={Lv, Qin and Josephson, William and Wang, Zhe and Charikar, Moses and Li, Kai}, language={en} }

@article{Andoni_Indyk_Laarhoven_Razenshteyn_Schmidt_2015, title={Practical and Optimal LSH for Angular Distance}, url={http://arxiv.org/abs/1509.02897}, abstractNote={We show the existence of a Locality-Sensitive Hashing (LSH) family for the angular distance that yields an approximate Near Neighbor Search algorithm with the asymptotically optimal running time exponent. Unlike earlier algorithms with this property (e.g., Spherical LSH [1, 2]), our algorithm is also practical, improving upon the well-studied hyperplane LSH [3] in practice. We also introduce a multiprobe version of this algorithm, and conduct experimental evaluation on real and synthetic data sets. We complement the above positive results with a ﬁne-grained lower bound for the quality of any LSH family for angular distance. Our lower bound implies that the above LSH family exhibits a trade-oﬀ between evaluation time and quality that is close to optimal for a natural class of LSH functions.}, note={arXiv:1509.02897 [cs]}, number={arXiv:1509.02897}, publisher={arXiv}, author={Andoni, Alexandr and Indyk, Piotr and Laarhoven, Thijs and Razenshteyn, Ilya and Schmidt, Ludwig}, year={2015}, month=sept, language={en} }

@inproceedings{Indyk_Motwani_1998, title={Approximate nearest neighbors: towards removing the curse of dimensionality}, ISBN={978-0-89791-962-3}, booktitle={Proceedings of the thirtieth annual ACM symposium on Theory of computing  - STOC ’98}, publisher={ACM Press}, author={Indyk, Piotr and Motwani, Rajeev}, year={1998}, pages={604–613}}

@inproceedings{mirflickr,
author = {Huiskes, Mark J. and Lew, Michael S.},
title = {The MIR flickr retrieval evaluation},
year = {2008},
isbn = {9781605583129},
publisher = {Association for Computing Machinery},
address = {New York, NY, USA},
booktitle = {Proceedings of the 1st ACM International Conference on Multimedia Information Retrieval},
pages = {39–43},
numpages = {5},
location = {Vancouver, British Columbia, Canada},
series = {MIR '08}
}

@misc{dino2,
      title={DINOv2: Learning Robust Visual Features without Supervision}, 
      author={Maxime Oquab and Timothée Darcet and Théo Moutakanni and Huy Vo and Marc Szafraniec and Vasil Khalidov and et al.},
      year={2023},
      url = {https://arxiv.org/abs/2304.07193},
      eprint={2304.07193},
      archivePrefix={arXiv},
      primaryClass={cs.CV}
}

@inproceedings{glove,
  author = {Jeffrey Pennington and Richard Socher and Christopher D. Manning},
  booktitle = {Empirical Methods in Natural Language Processing (EMNLP)},
  title = {GloVe: Global Vectors for Word Representation},
  year = {2014},
  pages = {1532--1543},
}

@article{gooaq2021,
  title={GooAQ: Open Question Answering with Diverse Answer Types},
  author={Khashabi, Daniel and Ng, Amos and Khot, Tushar and Sabharwal, Ashish and Hajishirzi, Hannaneh and Callison-Burch, Chris},
  journal={arXiv preprint},
  year={2021}
}

@inproceedings{sisap25, address={Cham}, title={Learning to Find Good Hash Functions for Embeddings}, ISBN={978-3-032-06069-3}, abstractNote={Our context of interest is locality sensitive hashing over embeddings generated by neural networks, in particular binary locality sensitive functions. These functions map objects in the embedding space to a binary value. To be locality sensitive, similar objects must have a higher probability of generating a hash collision than randomly selected objects. In this paper, we investigate the use of locality sensitive hash functions to address the ANN problem. We demonstrate that whilst uniform spaces exhibit good locality sensitivity properties, spaces derived from the output of neural networks do not. We describe a method to dynamically select instances of good locality sensitive hash functions on a per-query basis, based on a sample of the dataset. Using this analysis, we demonstrate a mechanism to efficiently search embedding spaces with linear preprocessing cost, allowing fast build times for large datasets.}, booktitle={Similarity Search and Applications}, publisher={Springer Nature Switzerland}, author={Claydon, Ben and Connor, Richard and Dearle, Alan}, year={2026}, pages={345–353} }

@book{muirhead2005aspects,
  author = {Muirhead, Robb J.},
  publisher = {Wiley-Interscience},
  title = {Aspects of Multivariate Statistical Theory},
  year = 2005
}

@inproceedings{Levina_Bickel, title={Maximum Likelihood Estimation of Intrinsic Dimension}, volume={17}, url={https://proceedings.neurips.cc/paper_files/paper/2004/file/74934548253bcab8490ebd74afed7031-Paper.pdf}, booktitle={Advances in Neural Information Processing Systems}, publisher={MIT Press}, author={Levina, Elizaveta and Bickel, Peter}, editor={Saul, L. and Weiss, Y. and Bottou, L.}, year={2004} }

@article{Muller_1959, title={A note on a method for generating points uniformly on n-dimensional spheres}, volume={2}, ISSN={0001-0782}, DOI={10.1145/377939.377946}, number={4}, journal={Communications of the ACM}, author={Muller, Mervin E.}, year={1959}, month=apr, pages={19–20} }

@article{Zhou_Huang_2003, title={Relevance feedback in image retrieval: A comprehensive review}, volume={8}, rights={http://www.springer.com/tdm}, ISSN={0942-4962, 1432-1882}, DOI={10.1007/s00530-002-0070-3}, abstractNote={We analyze the nature of the relevance feedback problem in a continuous representation space in the context of content-based image retrieval. Emphasis is put on exploring the uniqueness of the problem and comparing the assumptions, implementations, and merits of various solutions in the literature. An attempt is made to compile a list of critical issues to consider when designing a relevance feedback algorithm. With a comprehensive review as the main portion, this paper also offers some novel solutions and perspectives throughout the discussion.}, number={6}, journal={Multimedia Systems}, author={Zhou, Xiang Sean and Huang, Thomas S.}, year={2003}, month=apr, pages={536–544}, language={en} }

@article{rocchio1971relevance,
  title={Relevance feedback in information retrieval},
  author={Rocchio Jr, Joseph John},
  journal={The SMART retrieval system: experiments in automatic document processing},
  year={1971},
  publisher={Englewood Cliffs}
}

@inproceedings{Vadicamo_Scotti_Dearle_Connor_2025, address={Singapore}, title={Comparative Analysis of Relevance Feedback Techniques for Image Retrieval}, ISBN={978-981-96-2054-8}, DOI={10.1007/978-981-96-2054-8_16}, abstractNote={Relevance feedback mechanisms have garnered significant attention in content-based image and video retrieval thanks to their effectiveness in refining search results to better meet user information needs. This paper provides a comprehensive comparative analysis of four techniques: Rocchio, PicHunter, Polyadic Query, and linear Support Vector Machines, representing diverse strategies encompassing query vector modification, relevance probability estimation, adaptive similarity metrics, and classifier learning. We conducted experiments within an interactive image retrieval system, with varying amounts of user feedback: full feedback, limited positive feedback, and mixed feedback. In particular, we introduce novel enhanced versions of PicHunter and Polyadic search incorporating negative feedback. Our findings highlight the benefits of integrating both positive and negative examples, demonstrating significant performance improvements. Overall, SVM and our improved PicHunter outperformed the other approaches for ad-hoc search, especially in cases in which the feedback process is iterated several times.}, booktitle={MultiMedia Modeling}, publisher={Springer Nature}, author={Vadicamo, Lucia and Scotti, Francesca and Dearle, Alan and Connor, Richard}, editor={Ide, Ichiro and Kompatsiaris, Ioannis and Xu, Changsheng and Yanai, Keiji and Chu, Wei-Ta and Nitta, Naoko and Riegler, Michael and Yamasaki, Toshihiko}, year={2025}, pages={206–219}, language={en} }

@inproceedings{Claydon_Connor_Dearle_Vadicamo_2025, address={Cham}, title={Demonstrating the Efficacy of Polyadic Queries}, ISBN={978-3-031-75823-2}, abstractNote={Similarity search is normally defined to be the task of identifying those objects, from a large collection, that are most similar to a further single object presented as a query. Using polyadic queries, a small set of objects are presented to the system, with the intent of finding those objects most similar to all elements of the query set. A few scenarios have previously demonstrated the usefulness of this notion. For example, we may be searching for images similar to a red balloon over a lake. With a single query, it is impossible to tell if the intent is to search for other images of balloons over lakes, or for other red balloons in any background. If instead we could present a system with a few different images of balloons, all of which are either all red, or all over lakes, the similarity search engine may be able to respond more appropriately. In this paper we demonstrate software which permits the user to provide explicit feedback by selecting the best few results from an intermediate set which are best suited to their original information need. A polyadic query can be formed from this set, which should give better results with a minimum of user interaction.}, booktitle={Similarity Search and Applications}, publisher={Springer Nature Switzerland}, author={Claydon, Ben and Connor, Richard and Dearle, Alan and Vadicamo, Lucia}, year={2025}, pages={49–56} }

@inproceedings{word_expansion,
author = {Kuzi, Saar and Shtok, Anna and Kurland, Oren},
title = {Query Expansion Using Word Embeddings},
year = {2016},
isbn = {9781450340731},
publisher = {Association for Computing Machinery},
address = {New York, NY, USA},
url = {https://doi.org/10.1145/2983323.2983876},
doi = {10.1145/2983323.2983876},
booktitle = {Proceedings of the 25th ACM International on Conference on Information and Knowledge Management},
pages = {1929–1932},
numpages = {4},
location = {Indianapolis, Indiana, USA},
series = {CIKM '16}
}

@article{Xu_Croft_2000, title={Improving the effectiveness of information retrieval with local context analysis}, volume={18}, ISSN={1046-8188}, DOI={10.1145/333135.333138}, abstractNote={Techniques for automatic query expansion have been extensively studied in information research as a means of addressing the word mismatch between queries and documents. These techniques can be categorized as either global or local. While global techniques rely on analysis of a whole collection to discover word relationships, local techniques emphasize analysis of the top-ranked documents retrieved for a query. While local techniques have shown to be more effective that global techniques in general, existing local techniques are not robust and can seriously hurt retrieved when few of the retrieval documents are relevant. We propose a new technique, called local context analysis, which selects expansion terms based on cooccurrence with the query terms within the top-ranked documents. Experiments on a number of collections, both English and non-English, show that local context analysis offers more effective and consistent retrieval results.}, number={1}, journal={ACM Trans. Inf. Syst.}, author={Xu, Jinxi and Croft, W. Bruce}, year={2000}, month=jan, pages={79–112} }

@article{Mikolov_Chen_Corrado_Dean_2013, title={Efficient Estimation of Word Representations in Vector Space}, url={http://arxiv.org/abs/1301.3781}, DOI={10.48550/arXiv.1301.3781}, abstractNote={We propose two novel model architectures for computing continuous vector representations of words from very large data sets. The quality of these representations is measured in a word similarity task, and the results are compared to the previously best performing techniques based on different types of neural networks. We observe large improvements in accuracy at much lower computational cost, i.e. it takes less than a day to learn high quality word vectors from a 1.6 billion words data set. Furthermore, we show that these vectors provide state-of-the-art performance on our test set for measuring syntactic and semantic word similarities.}, note={arXiv:1301.3781 [cs]}, number={arXiv:1301.3781}, publisher={arXiv}, author={Mikolov, Tomas and Chen, Kai and Corrado, Greg and Dean, Jeffrey}, year={2013}, month=sept }

@article{searching_metric_spaces, title={Searching in metric spaces}, volume={33}, ISSN={0360-0300, 1557-7341}, DOI={10.1145/502807.502808}, abstractNote={The problem of searching the elements of a set that are close to a given query element under some similarity criterion has a vast number of applications in many branches of computer science, from pattern recognition to textual and multimedia information retrieval. We are interested in the rather general case where the similarity criterion defines a metric space, instead of the more restricted case of a vector space. Many solutions have been proposed in different areas, in many cases without cross-knowledge. Because of this, the same ideas have been reconceived several times, and very different presentations have been given for the same approaches. We present some basic results that explain the intrinsic difficulty of the search problem. This includes a quantitative definition of the elusive concept of “intrinsic dimensionality.” We also present a unified view of all the known proposals to organize metric spaces, so as to be able to understand them under a common framework. Most approaches turn out to be variations on a few different concepts. We organize those works in a taxonomy that allows us to devise new algorithms from combinations of concepts not noticed before because of the lack of communication between different communities. We present experiments validating our results and comparing the existing approaches. We finish with recommendations for practitioners and open questions for future development.}, number={3}, journal={ACM Computing Surveys}, author={Chávez, Edgar and Navarro, Gonzalo and Baeza-Yates, Ricardo and Marroquín, José Luis}, year={2001}, month=sept, pages={273–321}, language={en} }

\newpage
\appendix
\section{Description of Datasets}
\label{sec:datasets_appendix}

We use four high-dimensional vector datasets of different modalities for experimentation:

\begin{enumerate}
    \item MirFlickr is a collection of 1 million images submitted by Flickr users \cite{mirflickr}. We encode each of these images with the DinoV2S embedder \cite{dino2}, which outputs 384-dimensional embeddings of each image in Euclidean space. The dissimilarity metric for this space is either Euclidean or Cosine distance. Downloadable from \url{https://zenodo.org/records/15373201}.
    \item Twitter GloVe is a set of approximately 1.1 million word embeddings \cite{glove}. These embeddings were generated by analysing word co-occurrences from Twitter posts. The dissimilarity metric for this space is either Euclidean or Cosine distance. Downloadable from \url{https://nlp.stanford.edu/projects/glove/}.
    \item GOOAQ is a set of 3 million questions submitted to the Google search engine and their answers \cite{gooaq2021}. They were embedded using a sentence BERT model to produce 384 dimensional feature vectors.%
    Downloadable at \url{https://huggingface.co/datasets/sentence-transformers/gooaq}. %
    \item A uniform spherical dataset of 1 million points in 200 dimensions. Generated using the method described by Muller \cite{Muller_1959}.
\end{enumerate}

\section{Statistical Model of Hyperplane Hash Function}
\label{sec:statistical_model}

We derive a statistical model which describes the probability that a point $\mathbf{u} \in \mathbb{R}^d \; : \; ||\mathbf{u}||_2 = 1 $ and each member of some set of points $knn(\mathbf{q})$ will hash collide under the family $\mathcal{H}_{rp}$. This class of hash function is parametrised by a random hyperplane $\mathbf{w}$ passing through the origin in a random direction, which can be generated by a vector of numbers drawn from a standard Gaussian distribution. Any points which are at a greater height above $\mathbf{w}$ than a uniformly generated value $a \sim \mathcal{U}[min(X \mathbf{w}^T), max(X \mathbf{w}^T)]$ are assigned a value of $1$, and all others assigned a value of $0$. 

To build a model of how this set of hash functions will interact with a point $\mathbf{u}$ and a set $\knnq$, we begin with the distribution which describes the set of all normals to planes which parametrise all hash functions, as in Equation \ref{eq:w_distribution}. This is a standard multivariate Gaussian distribution, which is the co-ordinate-wise distribution for generating random hyperplane normals. 

\begin{equation}
    \mathbf{w} \sim \mathcal{N}\left( \mathbf{0}, \mathbf{\mathbf{I}} \right)
    \label{eq:w_distribution}
\end{equation}

We may condition the distribution of $\mathbf{w}$ given that $\mathbf{u} \cdot \mathbf{w} = b$, where $b \in \mathbb{R}$ represents the height of $\mathbf{u}$ above $\mathbf{w}$. The resultant distribution describes the subset of all planes which both pass through the origin and where $\mathbf{u}$ has a fixed dot product to $\mathbf{w}$, thereby fixing the height of $\mathbf{u}$ above $\mathbf{w}$. 
A formal derivation of conditional Gaussian distributions can be found in \cite{muirhead2005aspects}.

\begin{equation}
    \mathbf{w} \; | \; \mathbf{u} \cdot \mathbf{w} = b \sim \mathcal{N}(b\mathbf{u}, \mathbf{I} - \mathbf{u}\mathbf{u}^T)
    \label{eq:w_distribution_conditioned}
\end{equation}

Finally, we must model how this distribution of planes interacts with the set $knn(\mathbf{q})$. Consider the set $knn(\mathbf{q})$ as a $k \times d$ matrix $\mathbf{A}$, where the $i$th row is the $i$th near neighbour of the original query, $\mathbf{q}$. The matrix multiplication $\mathbf{Aw}$ will produce a vector of the form $[\mathbf{x}_0 \cdot \mathbf{w}, \mathbf{x}_1 \cdot \mathbf{w}, \cdots, x_{k-1} \cdot \mathbf{w}]$ where $\mathbf{x}_i$ is the $i$th near neighbour.

\begin{align}
\mathbf{A} \mathbf{w} \; | \; \mathbf{u} \cdot \mathbf{w} = b 
&\sim \mathcal{N}\big(\mathbf{A}b\mathbf{u},\, \mathbf{A}(\mathbf{I} - \mathbf{u}\mathbf{u}^{T})\mathbf{A}^{T}\big) \\
&\sim \mathcal{N}\big(\mathbf{A}b\mathbf{u},\, \mathbf{A}\mathbf{A}^{T} - (\mathbf{A}\mathbf{u})(\mathbf{A}\mathbf{u})^{T}\big)
\label{eq:w_dist}
\end{align}

We may alternatively write this distribution in terms of dot products.
\begin{equation}
    \mathbf{A} \mathbf{w} \; | \; \mathbf{u} \cdot \mathbf{w} = b \sim \mathcal{N} \left(\mu_i = b (\mathbf{u} \cdot \mathbf{x}_i)
, \Sigma_{ij}=\mathbf{x}_i\cdot\mathbf{x}_j-(\mathbf{x}_i\cdot\mathbf{u})(\mathbf{x}_j\cdot\mathbf{u})
\right), \ 0 \leq i,j < k
\label{eq:hp_normal_form}
\end{equation}

This multivariate Gaussian distribution models the mean and covariance of the values $\mathbf{x}_i \cdot \mathbf{w}$ for all $\mathbf{x}_i \in \knnq$. The hyperplane hash function may be modelled by drawing and thresholding on a random value $a$, thus creating a binary output:

\begin{align}
\mathbf{A} \mathbf{w} \; | \; \mathbf{u} \cdot \mathbf{w} = b \sim \mathcal{N}\big(\mathbf{A}b\mathbf{u},\, \mathbf{A}\mathbf{A}^{T} - (\mathbf{A}\mathbf{u})(\mathbf{A}\mathbf{u})^{T}\big)  > a \mathbf{1}_k
\label{eq:binary_lsh_matrix}
\end{align}

This models the inner products between $\mathbf{A}$ and $\mathbf{w}$. However, there is not an exact hash collision probability formula for $\mathcal{H}_{rp}$. Instead, our strategy will be to show that the hash collision probability for $\mathcal{H}_{rp}$ is similar to $\mathcal{H}_{c}$ in Section \ref{sec:proof_a_is_0}. We may then consider only the closed-form collision probability of $\mathcal{H}_c$ thereafter. Using this formula, we will show that $\normc$ is a good instance of $\mathbf{u}$ by proving that it maximises the average collision probability, to a first-order approximation.

\subsection{Approximating the Effect of $a$}
\label{sec:proof_a_is_0}

Consider Equation \ref{eq:binary_lsh_matrix}. We may rewrite this expression as a convolution of a normal and 1-dimensional uniform distribution by subtracting $a$ from both sides of the inequality:

\[
\mathcal{N}\big(\mathbf{A}b\mathbf{u},\, \mathbf{A}\mathbf{A}^{T} - (\mathbf{A}\mathbf{u})(\mathbf{A}\mathbf{u})^{T}\big) - a \mathbf{1}_k > 0
\]

More precisely, $a$ is distributed as a degenerate $k$-dimensional uniform distribution, whereby a single, fixed, uniform scalar value is drawn between $min(X \mathbf{w}^T), max(X \mathbf{w}^T)$ (where each $X$ is the set of points being hashed, thereby making these bounds the height of the furthest points in $X$ above and below $\mathbf{w}$), which is repeated for all dimensions. We denote this distribution $\mathcal{U}$.  

\[
\mathcal{N}\big(\mathbf{A}b\mathbf{u},\, \mathbf{A}\mathbf{A}^{T} - (\mathbf{A}\mathbf{u})(\mathbf{A}\mathbf{u})^{T}\big) - \mathcal{U} \left(min(X \mathbf{w}^T), max(X \mathbf{w}^T) \right) > 0
\]

For ease of notation, we let $W \sim \mathcal{N}\big(\mathbf{A}b\mathbf{u},\, \mathbf{A}\mathbf{A}^{T} - (\mathbf{A}\mathbf{u})(\mathbf{A}\mathbf{u})^{T}\big)$ and $Y \sim \mathcal{U} \left(min(X \mathbf{w}^T), max(X \mathbf{w}^T) \right)$. We further define a new distribution $Z = W - Y$. 

Maximising ACP is equivalent to maximising the probability of the each co-ordinate of $Z$ having the same sign as $b - a$, as the hash of both $\mathbf{u}$ and $\mathbf{x}_i$ will be equal in this case. However, this property is very hard to prove for $Z$ itself, due to the complexity of the distribution%
\footnote{The author did derive the distribution of $Z$, but there is little to be gained from recounting it.}. However, we show that finding a $\mathbf{u}$ which maximises the expected value of $W$ also implies that the $\mathbf{u}$ maximises the expected value of $Z$. Similarly, the $\mathbf{u}$ which minimises the variance of $W$ directly implies it minimises the variance of $Z$. By showing this, we may work with the much more wieldy $X$ to prove properties of $Z$, greatly simplifying our derivation.

To achieve this, we must first assume that $max(X \mathbf{w}^T) \approx -min(X \mathbf{w}^T)$. This is because the distribution of $ X \mathbf{w}^T$ is itself Gaussian, if each member of $X$ is reasonably decorrelated%
\footnote{This should be true, given a correct level of dimensionality reduction has been performed on the data.}. We may use this fact to derive a simple mean and variance of the degenerate uniform distribution $Y$:

\[
\mathbb{E}[Y] = 0
\]

\[
\Sigma[Y] = \frac{max(\mathbf{x_i} \cdot \mathbf{w})^2}{3}\mathbf{J}_k
\]

where $\mathbf{J}_k$ is the $k \times k$ matrix of ones. To compute the mean and variance of the distribution $Z$, as we assume $X$ and $Y$ are independent (which is reasonable, as the maximum and minimum of a large collection of Gaussian vectors is approximately constant). We can thus rewrite the mean and variance of $Z$ in terms of $W$ and $Y$:

\begin{align*}
    \mathbb{E}[Z] &= \mathbb{E}[W] -\mathbb{E}[Y]\\
    \implies \mathbb{E}[Z] &= \mathbb{E}[W]
\end{align*}

\[
\Sigma[Z] = \Sigma[W] + \frac{max(\mathbf{x_i} \cdot \mathbf{w})^2}{3}\mathbf{J}_k
\]

As the expected values of $W$ and $Z$ are equal, and the covariance matrix of $Z$ is simply the covariance matrix of $X$ plus some constant, it follows that $max(\mathbb{E}[W]) \implies max(\mathbb{E}[Z])$, and $min(Var[W]) \implies min(Var[Z])$. Thus we work only with the much simpler distribution:

\begin{equation}
    W = 
    \mathbf{A} \mathbf{w} \; | \; \mathbf{u} \cdot \mathbf{w} = b \sim \mathcal{N}\big(\mathbf{A}b\mathbf{u},\, \mathbf{A}\mathbf{A}^{T} - (\mathbf{A}\mathbf{u}) (\mathbf{A}\mathbf{u})^{T}\big) > 0
    \label{eq:X}
\end{equation}

Moreover, proving this relationship allows us to simultaneously demonstrate properties of both Charikar's binary hash function (where $a=0$), and the hash function used by RP-Forest (which takes a uniform value of $a$), as the value $a$ has been elided from our analysis. We emphasise that the results derived for $\mathcal{H}_{rp}$ are approximations under the stated independence and symmetry assumptions.

\section{$\normc$ Increases Mean Collision Probability}
\label{sec:proofs}

We prove that setting $\mathbf{u} = \normc$, to a stated degree of error, maximises the probability of the hashes of $\mathbf{u}$ and some $\mathbf{x}_i \in knn(\mathbf{q})$ colliding. This proof is split into two main parts: Theorem \ref{thm:max_e} shows that setting $\mathbf{u} = \normc$ maximises the expected value of the distribution derived in Equation \ref{eq:X} if $b > 0$, and minimises it if $b < 0$; Theorem \ref{thm:collision_dot} shows that maximising this expected value also maximises the average collision probability to a first order approximation. In Section \ref{sec:avg_covariance}, we prove that the variance of the distribution in Equation \ref{eq:X} is minimised if $\mathbf{u} = \normc$. This pair of properties is sufficient to motivate that $\normc$ is highly likely to hash collide with members of $\knnq$.

\begin{theorem}
\label{thm:max_e}
    If $\mathbf{u} = \normc$, $\mathbb{E}[\sum_i(\mathbf{w} \cdot \mathbf{x_i} \; | \; \mathbf{u} \cdot \mathbf{w} = b)]$ is maximised if $b > 0$, and minimised if $b < 0$.
\end{theorem}

\begin{proof}
    As our domain of interest is $\ell_2$-normed spaces, we may rewrite the expectation of Equation \ref{eq:hp_normal_form} using the polarisation identity:
    \[
    b \cdot
    \begin{pmatrix}
    1 - \frac{1}{2}||\mathbf{u} - \mathbf{x}_1||^2_2 \\
    1 - \frac{1}{2}||\mathbf{u} - \mathbf{x}_2||^2_2  \\
    \vdots \\
    1 - \frac{1}{2}||\mathbf{u} - \mathbf{x}_k||^2_2 
    \end{pmatrix}
    \]

To maximise the sum of all elements, we may show that:

\begin{align*}
    &\argmax_{||\mathbf{u}||_2 = 1} \; 1 - \frac{1}{2}\sum_{i=0}^{k-1} ||\mathbf{u} -\mathbf{x}_i||_2^2\\
    &=\argmin_{||\mathbf{u}||_2 = 1} \sum_{i=0}^{k-1} ||\mathbf{u} -\mathbf{x}_i||_2^2\\
    &=\argmin_{||\mathbf{u}||_2 = 1} \sum_{i=0}^{k-1} 2 - 2(\mathbf{u} \cdot \mathbf{x_i})\\
    &=\argmin_{||\mathbf{u}||_2 = 1} 2k - 2\sum_{i=0}^{k-1}(\mathbf{u} \cdot \mathbf{x_i})\\
    &=\argmin_{||\mathbf{u}||_2 = 1} 2k - 2 \mathbf{u} \cdot  \sum_{i=0}^{k-1}\mathbf{x_i}\\
    &=\argmax_{||\mathbf{u}||_2 = 1} \mathbf{u} \cdot \sum_{i=0}^{k-1}\mathbf{x_i}
\end{align*}

We must now choose the vector $\mathbf{u}$ which maximises this inner product under restriction that $||\mathbf{u}||_2=1$. To begin, we apply the Cauchy-Schwarz inequality:

\begin{align*}
    \mathbf{u} \cdot \sum_{i=0}^{k-1}\mathbf{x_i} &\leq ||\mathbf{u}||_2 ||\sum_{i=0}^{k-1}\mathbf{x_i}||_2\\
    \mathbf{u} \cdot \sum_{i=0}^{k-1}\mathbf{x_i} &\leq ||\sum_{i=0}^{k-1}\mathbf{x_i}||_2
\end{align*}

Thus, the maximum of the objective function $\max_{||\mathbf{u}||_2 = 1} \mathbf{u} \cdot \sum_{i=0}^{k-1}\mathbf{x_i}$, for the best possible $\mathbf{u}$, is $||\sum_{i=0}^{k-1}\mathbf{x_i}||_2$.

Substituting $\mathbf{u} = \normc$ into the objective function, we show that this point attains that maximum value:

\begin{align*}
    &= \frac{\mathbf{c}}{||\mathbf{c}||_2} \cdot \sum_{i=0}^{k-1}\mathbf{x_i}\\
    &= \frac{\sum_{i=0}^{k-1}\mathbf{x_i}}{k||\frac{1}{k}\sum_{i=0}^{k-1}\mathbf{x_i}||_2} \cdot \sum_{i=0}^{k-1}\mathbf{x_i}\\
    &= \frac{\sum_{i=0}^{k-1}\mathbf{x_i}}{||\sum_{i=0}^{k-1}\mathbf{x_i}||_2} \cdot \sum_{i=0}^{k-1}\mathbf{x_i}\\
    &= \frac{(\sum_{i=0}^{k-1}\mathbf{x_i}) \cdot (\sum_{i=0}^{k-1}\mathbf{x_i})}{||\sum_{i=0}^{k-1}\mathbf{x_i}||_2}\\
    &= \frac{||\sum_{i=0}^{k-1}\mathbf{x_i}||^2_2}{||\sum_{i=0}^{k-1}\mathbf{x_i}||_2}\\
    &= ||\sum_{i=0}^{k-1}\mathbf{x_i}||_2
\end{align*}

and thus $\normc$ is the $\ell_2$ normalised vector which maximises the expected value of Equation \ref{eq:hp_normal_form} as required.

\end{proof}

\begin{theorem}
    $\normc$ is the unit vector which maximises the first-order approximation of the average collision probability.
    \label{thm:collision_dot}
\end{theorem}

\begin{proof}
    The collision probability between any two $\ell_2$-normalised vectors, for Charikar's hash function family $\mathcal{H}_{c}$\cite{Charikar_2002}, is:

    \[
    \Pr(h(\mathbf{a}) = h(\mathbf{b})) = 1 - \frac{cos^{-1}(\mathbf{a} \cdot \mathbf{b})}{\pi}
    \]

    Therefore, to maximise the average collision probability, we must maximise the objective function:

    \[
    \argmax_{||\mathbf{u}||_2 = 1} \frac{1}{k} \sum_{i=0}^{k-1} 1 - \frac{cos^{-1}( \mathbf{x}_i \cdot \mathbf{u})}{\pi}
    \]

    We approximate the function $\cos^{-1}$ with its first-order Maclaurin expansion.

    \[
    = \argmax_{||\mathbf{u}||_2 = 1} \frac{1}{k} \sum_{i=0}^{k-1} 1 - \frac{\frac{\pi}{2} -  \mathbf{x}_i \cdot \mathbf{u} + \mathcal{O}((\mathbf{x}_i \cdot \mathbf{u})^3)}{\pi}
    \]

    Simplifying:

    \[
    = \argmax_{||\mathbf{u}||_2 = 1} \frac{1}{2}
    +
    \frac{1}{k\pi}\sum_{i=0}^{k-1}\mathbf{x}_i \cdot \mathbf{u}
    +
    \mathcal{O}\!\left(
    \frac{1}{k}\sum_{i=0}^{k-1}
    (\mathbf{x}_i \cdot \mathbf{u})^3
    \right)
    \]
    
    By Theorem \ref{thm:max_e}, we have already shown that the term $\sum_{i=0}^{k-1}\mathbf{x}_i \cdot \mathbf{u}$ is maximised when $\mathbf{u} = \normc$. Therefore, using $\normc$ maximises the first-order approximation:

    \[
    \normc = \argmax_{||\mathbf{u}||_2 = 1}
    \frac{1}{2}
    +
    \frac{1}{k\pi}\sum_{i=0}^{k-1}\mathbf{x}_i \cdot \mathbf{u}
    \]
\end{proof}

\subsubsection{Discussion}

Upon first inspection, it may seem that a first-order approximation is insufficient to model the function $cos^{-1}$. However, it is accurate for lower values of the range of the function $[0, 1]$. This is demonstrated by Figure \ref{fig:approximation_vs_collision_prob} which shows the exact collision probability of two vectors with given inner product, and the first order approximation of this probability.

\begin{figure}[h]
    \centering
    \includegraphics[width=0.6\linewidth]{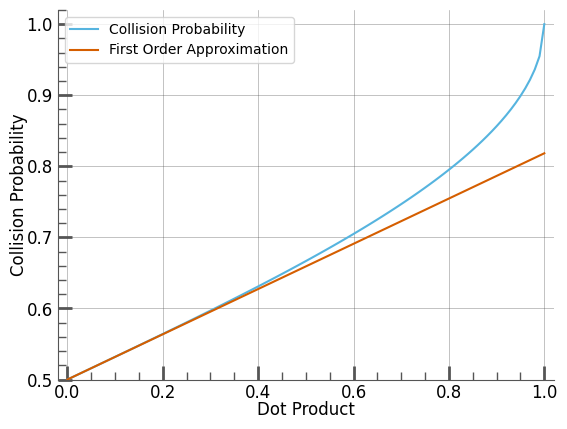}
    \caption{The hash collision probability ($y$ axis) of two $\ell_2$-normalised vectors with a given inner product ($x$ axis).}
    \label{fig:approximation_vs_collision_prob}
\end{figure}

From inspection, the first order approximation has little error as $\mathbf{x}_i \cdot \mathbf{u}$ is small. As $\mathbf{x}_i \cdot \mathbf{u} \rightarrow 1$, the error increases. However, there are two reasons that this increased error does not trouble us:

\begin{enumerate}
    \item If there exists a point $\mathbf{v}$ which is very close to $\mathbf{u}$, then its true collision probability will be underestimated by the first-order approximation. Therefore, it will not feature in the sum as prominently as it should.   However, their inner product will still be large even when this error is accounted for, and therefore it is likely that $\mathbf{v}$ will quickly be retrieved.
    \item Definitionally, the distribution of distances between points in a high-dimensional space has a high mean and low variance. Therefore, we expect few examples of high $\mathbf{u} \cdot \mathbf{v}$ \cite{searching_metric_spaces} in practice. 
\end{enumerate}

\subsection{$\normc$ Minimises Average Covariance}
\label{sec:avg_covariance}


    



\label{sec:off_diag_proof}

Here, we use the matrix $\Sigma$ as a proxy through which to analyse the output correlation binary hash functions as applied to a fixed pair of points. For a fixed pair of points $\mathbf{x}_i, \mathbf{x}_j$, a matrix entry $\Sigma_{i, j}$ value of $-1$ represents that, for all $h \in \mathcal{H}_{c}$, $h(\mathbf{x_i}) \neq h(\mathbf{x}_j)$, and a correlation of $+1$ represents the converse. A score of $0$ implies that the outcomes of the hash functions are decorrelated, and that the value $h(\mathbf{x}_i)$ provides no information on the value $h(\mathbf{x}_j)$. Positive covariance matrix entries imply a positive correlation of hash function outcomes, and negative covariance matrix entries imply a negative outcome. However, the relationship between the value of $\Sigma_{i, j}$ and the observed $\operatorname{Cov}(h(\mathbf{x}_i), h(\mathbf{x}_j))$ may not be linear.

We define the \emph{average correlation} more formally to be the mean value of $\Sigma_{i, j}$ for all  $i, j$. If the average correlation is high, then it is likely that many or all members of $knn(\mathbf{q})$ will have equal hash function output if a single function from $\mathcal{H}_{c}$ is applied. However, if $\mathbf{q}$ produces the opposite output, then hash failure occurs. To minimise the probability of hash failures, a $\mathbf{u}$ must be chosen which induces a low average correlation.

We show that using $\mathbf{u} = \normc$ minimises the average correlation:

\begin{proof}
    \begin{align*}
    &\argmin_{||\mathbf{u}||_2 = 1} \frac{1}{k^2} \sum_{i=0}^{k-1} \sum_{j=0}^{k-1} \mathbf{x}_i \cdot \mathbf{\mathbf{x}_j} - (\mathbf{x}_i \cdot \mathbf{u})(\mathbf{x}_j \cdot \mathbf{u}) \\
    &\texttt{Because $\mathbf{x}_i$, $\mathbf{x}_j$ are constant}\\
    &= \argmax_{||\mathbf{u}||_2 = 1} \sum_{i=0}^{k-1} \sum_{j=0}^{k - 1} (\mathbf{x}_i \cdot \mathbf{u})
    (\mathbf{x}_j \cdot \mathbf{u})\\
    &= \argmax_{||\mathbf{u}||_2 = 1} \left (\sum_{i=0}^{k-1} (\mathbf{x}_i \cdot \mathbf{u}) \right) \left (\sum_{i=0}^{k-1} (\mathbf{x}_i \cdot \mathbf{u}) \right)\\
    &= \argmax_{||\mathbf{u}||_2 = 1} ( \sum_{i=0}^{k-1} \mathbf{x}_i \cdot \mathbf{u})^2
\end{align*}
\end{proof}

As we showed in Theorem \ref{thm:max_e}, $\normc$ is the point which maximises the term $\sum_{i=0}^{k-1} \mathbf{x}_i \cdot \mathbf{u}$. Therefore, it also maximises $(\sum_{i=0}^{k-1} \mathbf{x}_i \cdot \mathbf{u})^2$. As $\Sigma$ is a proxy for describing how correlated hash functions are when applied to pairs of points, we expect that hash function correlation is similarly low as well.

\section{Distribution of Off-Diagonal Covariance Values}
\label{sec:off-diag}
We demonstrate that using $\mathbf{u} = \normc$ causes the \emph{average} off-diagonal entry of the covariance matrix to be approximately $0$. We measure covariance matrices for a large set of queries and record their distribution. We formulate this experiment as follows:

\begin{enumerate}
    \item For a given query, generate $1000$ random hyperplanes such that $|\mathbf{u} \cdot \mathbf{w}| < 0.005$ (approximating planes where $b=0$). Note that the choice of $b=0$ is arbitrary as the value of $b$ does not influence the covariance as per Equation \ref{eq:hp_normal_form}, but must be fixed such that the mean value is stable.
    \item For each plane $\mathbf{w}$, we compute the values $A\mathbf{w}$, or the height of each near neighbour above $\mathbf{w}$. We store these results in $D$, a $1000 \times 100$ matrix representing the heights of the $100$ nearest neighbours to the query over all $1000$ hyperplanes.
    \item Finally, we measure the covariance between the columns of the matrix $D$. This measures the degree of correlation between $\mathbf{x}_i \cdot \mathbf{w}$ and $\mathbf{x}_j \cdot \mathbf{w}$ over a large sample of planes. Recall that these covariances are bounded between $\pm 1$.
\end{enumerate}

Figure \ref{fig:covariance_hists} shows the distribution of all such covariances aggregated over all queries. For all datasets, setting $\mathbf{u} = \normc$ reduces the average covariance to approximately zero.

\begin{figure}
    \centering
    \begin{subfigure}{0.48\textwidth}
        \centering
        \includegraphics[width=\linewidth]{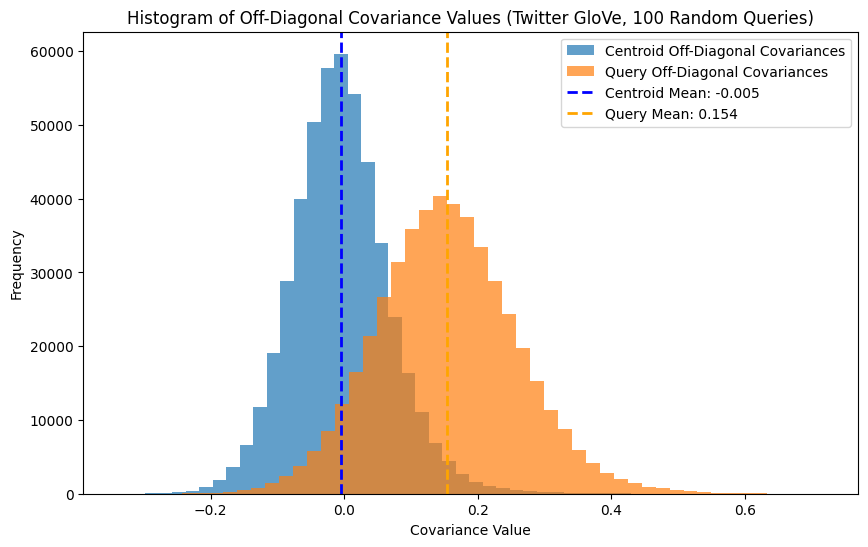}
        \caption{GloVe}
        \label{fig:cov_glove}
    \end{subfigure}
    \hfill
    \begin{subfigure}{0.48\textwidth}
        \centering
        \includegraphics[width=\linewidth]{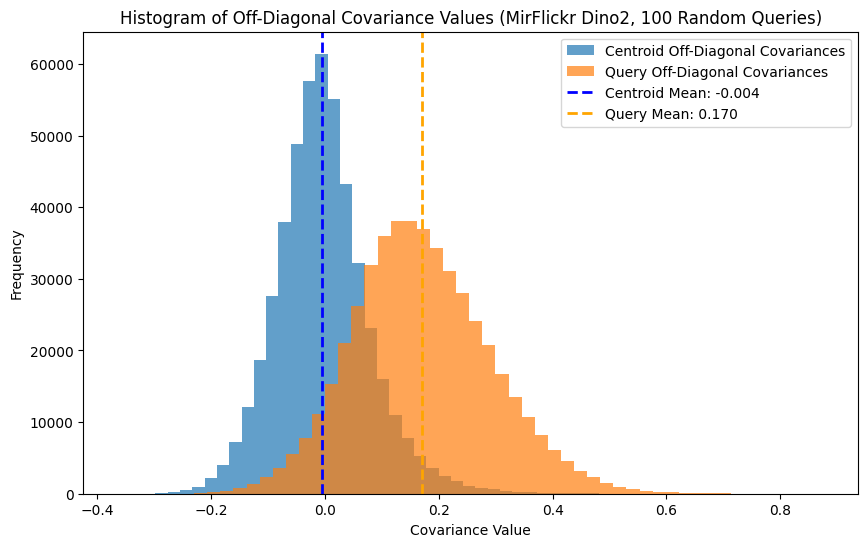}
        \caption{DINOv2}
        \label{fig:cov_dino2}
    \end{subfigure}

    \vspace{1em}
    \begin{subfigure}{0.48\textwidth}
        \centering
        \includegraphics[width=\linewidth]{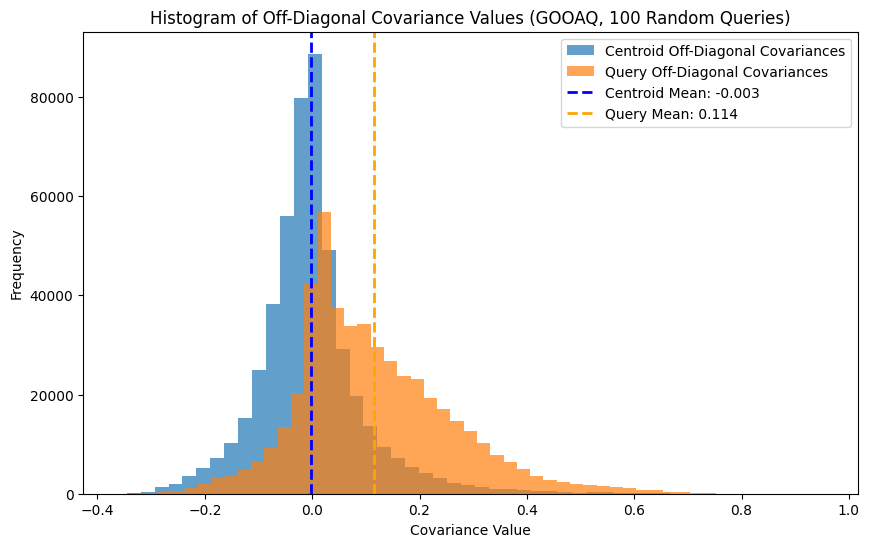}
        \caption{GooAQ}
        \label{fig:cov_gooaq}
    \end{subfigure}

    \caption{Aggregated off-diagonal covariance values when applying a large number of binary hash functions to $\normc$ (blue) and $\mathbf{q}$ (orange) as instances of $\mathbf{u}$. Measured over 100 random queries.}
    \label{fig:covariance_hists}
\end{figure}

\section{A Statistical Argument for $\mathbf{\chat}$}
\label{sec:clt_proof}
Here we present a statistical model which explains the observation that $||\mathbf{q} - \mathbf{c}||_2 > ||\mathbf{\chat} - \mathbf{c}||$ with a high probability. Recall that a smaller distance to $\mathbf{c}$ implies a better instance of $\chat$. Note that here we work with non-normalised vectors. In Section \ref{sec:statistic_discussion}, we discuss how the argument presented here should apply to the inequality $||\mathbf{q} - \normc||_2 > ||\normchat - \mathbf{c}||$ in non-pathological cases.

\subsection{Definitions and Preliminaries}
To begin, we define the terms:

\begin{itemize}
    \item $\boldsymbol{\mu} \in \mathbb{R}^d$: The user-provided query, previously denoted $\mathbf{q}$. We rename it to $\boldsymbol{\mu}$ here as we wish to emphasise the role of $\mathbf{q}$ as the mean value $\boldsymbol{\mu}$ of the local uniform model, and its role as the mean value as used by the central limit theorem.
    \item $\mathbf{x}_0 \cdots \mathbf{x}_{r-1} \in \mathbb{R}^d$ are the $r$ near neighbours of $\boldsymbol{\mu}$.
    \item $\mathbf{c} \in \mathbb{R}^d$ is the centroid of the $k$ near neighbour set of $\boldsymbol{\mu}$; $\{\mathbf{x}_0 \cdots \mathbf{x}_{k-1}\}$.
    
    \item $\mathbf{\chat} \in \mathbb{R}^d$ is the centroid of a random selection of $k$ objects from the set $\{\mathbf{x}_0 \cdots \mathbf{x}_{r-1}\}$ for some $r > k$. We denote this sample set $S$. We use $S$ to model an approximation of the set $knn(\mathbf{q})$ generated some coarse $k$-nearest-neighbour search.
\end{itemize}

We also assume that the set of points $\{\mathbf{x}_0 \cdots \mathbf{x}_{r-1} \}$ exist in some intrinsically $m$-dimensional space, and that the local density about $\mathbf{q}$, which we call $f(\mathbf{q})$, is approximately constant.  These are the same set of assumptions made by Levina and Bickel \cite{Levina_Bickel} in their seminal paper deriving the maximum likelihood estimator of local intrinsic dimensionality. We refer readers to this paper for a fuller treatment of local intrinsic dimensionality. A brief intuition is as follows: for some fixed point $\mathbf{x}$, the probability that some other point $\mathbf{y}$ exists in the dataset whereby $||\mathbf{x} - \mathbf{y}||_2 = t$ is $a V_m(t)$, where $V_m(t)$ is the volume of the hypersphere in $m$ dimensions of radius $t$, and $a$ is some normalising constant which is derived from the density $f(\mathbf{x})$. 

In the next sections, we present three theorems which will demonstrate that, under stated assumptions, $\Pr(||\mathbf{q} - \mathbf{c}||_2 > ||\mathbf{\chat} - \mathbf{c}||)$ is large. We begin in Theorem \ref{thm:clt}, in which we derive a normally distributed co-ordinate-wise approximation for the vectors $\mathbf{c} - \boldsymbol{\mu}$ and $\mathbf{\chat} - \boldsymbol{\mu}$ using the central limit theorem (CLT).

\begin{theorem}
\[
\cnobold_i - \mu_i \xrightarrow{d} \mathcal{N} \left(0, \sigma_1^2 \right)
\]

\[
\chatnobold_i - \mu_i \xrightarrow{d} \mathcal{N} \left(0, \sigma_2^2 \right)
\]

where $\sigma_1^2 \leq \sigma_2^2$
\label{thm:clt}
\end{theorem}

\begin{proof}

To begin, recall that the $i$th co-ordinate of $\mathbf{c}$ and $\mathbf{\chat}$ is defined by the centroid as:

\[
\cnobold_i = \frac{\sum_{j=0}^{k-1}x_{j, i}}{k}
\]

\[
\chatnobold_i = \frac{\sum_{j=0}^{k-1}S_{j, i}}{k}
\]

We have assumed that the near neighbours of $\boldsymbol{\mu}$ are sampled uniformly from the $m$-ball centred at $\boldsymbol{\mu}$. For each fixed co-ordinate $i$, we model the co-ordinate values of points sampled within this locality as independent and identically distributed random variables (i.i.d) with mean at $\mu_i$. We use the variables $\xi_1^2$ and $\xi_2^2$ to represent the (unknown) variances of the underlying $i$th co-ordinate distributions of points in $\knnmu$ and $S$ respectively. Under the assumptions of our model, $\xi_1^2 \leq \xi_2^2$ as the points $\mathbf{x}_0 \cdots \mathbf{x}_{k-1}$ are more tightly distributed about $\boldsymbol{\mu}$ than those points in $S$. Using this assumption, we may directly apply the central limit theorem to as follows:

\[
\sqrt{k}(\frac{\sum_{j=0}^{k-1}x_{j, i}}{k} - \mu_i) \xrightarrow{d} \mathcal{N}(0, \xi_1^2)
\]

\[
\sqrt{k}(\frac{\sum_{j=0}^{k-1}S_{j, i}}{k} - \mu_i) \xrightarrow{d} \mathcal{N}(0, \xi_2^2)
\]

Rearranging, we reach the distributions:

\[
(\cnobold_i - \mu_i) \xrightarrow{d} \mathcal{N} \left(0, \frac{\xi_1^2}{k} \right)
\]

\[
(\chatnobold_i - \mu_i) \xrightarrow{d} \mathcal{N} \left(0, \frac{\xi_2^2}{k} \right)
\]

We denote $\sigma_1^2 = \frac{\xi_1^2}{k}$, and $\sigma_2^2 = \frac{\xi_2^2}{k}$ respectively.
\end{proof}

In Theorem \ref{thm:variance_gap}, we use the results derived in Theorem \ref{thm:clt} to prove that the co-ordinates of the vector $(\mathbf{\chat} - \mathbf{c})$ are also normally distributed, with a smaller variance than the co-ordinates of $(\boldsymbol{\mu} - \mathbf{c})$. 

\label{sec:var_diff}
\begin{theorem}
    The co-ordinates $(\chatnobold_i - \cnobold_i)$ are also normally distributed, and $Var(\chatnobold_i - \cnobold_i) < Var(\cnobold_i - \mu_i)$, if the terms $\cnobold_i - \mu_i$ and $\chatnobold_i - \mu_i$ have sufficient covariance. 
    \label{thm:variance_gap}
\end{theorem}

\begin{proof}
    As $(\cnobold_i - \mu_i)$ and $(\chatnobold_i - \mu_i)$ are each normally distributed, we may compute their convolution as follows:

    \[
    (\cnobold_i - \mu_i) - (\chatnobold_i - \mu_i) \xrightarrow{d} \mathcal{N}(0, \sigma^2_1 + \sigma^2_2 - 2 Cov(\cnobold_i - \mu_i, \chatnobold_i - \mu_i))
    \]

    Simplifying the left hand side:

    \begin{equation}
    (\cnobold_i - \chatnobold_i) \xrightarrow{d} \mathcal{N}(0, \sigma^2_1 + \sigma^2_2 - 2 Cov(\cnobold_i - \mu_i, \chatnobold_i - \mu_i))
    \label{eq:c_ci_dist}
    \end{equation}
    
    Previously, we assumed that the points are distributed uniformly and proportionally to the $m$-ball centred on $\mu$. Therefore, we may model the sets $\knnmu$ and $S$ as each having been drawn from spherical uniform distributions, where the former has been drawn from an $m$-ball with a radius $r_1$, and the second drawn from an $m$-ball with a radius $r_2$. In this example, $r_1$ is the $k$-nearest neighbour threshold, and $r_2$ is the $r$-nearest neighbour threshold.
    
    A property of the spherical uniform distribution is that the variance of co-ordinates of points drawn uniformly from the $m$-ball is directly proportional to the square of the radius of the ball. Therefore, as both the distributions are uniform and have the same centre, we may write the $\sigma_1^2$ as a ratio of $\sigma_2^2$

    \[
    \sigma_2^2 = \frac{r_2^2}{r_1^2}\sigma_1^2
    \]
    
    Using this fact, we may rewrite all terms of the distribution $(\chatnobold_i - \cnobold_i)$ dependent on $(\chatnobold_i-\mu_i)$ in terms of $(\cnobold_i - \mu_i)$:

    \[
    (\chatnobold_i - \cnobold_i) \xrightarrow{d} \mathcal{N}(0, \sigma^2_1 + \left(\frac{r_2}{r_1} \right)^2 \sigma^2_1 - 2 Cov(\cnobold_i - \mu_i, (\chatnobold_i - \mu_i))
    \]

    We define the term $\rho = Cov(\cnobold_i - \mu_i, \chatnobold_i - \mu_i)$:

    \[
    (\chatnobold_i - \cnobold_i) \xrightarrow{d} \mathcal{N}(0, \sigma^2_1 + \left(\frac{r_2}{r_1} \right)^2 \sigma^2_1 - 2 \rho )
    \]

    In order for $Var(\chatnobold_i - \cnobold_i) < Var(\cnobold_i - \mu_i)$, we require:

    \[
    \sigma^2_1 + \left(\frac{r_2}{r_1} \right)^2 \sigma^2_1 - 2 \rho <  \sigma_1^2
    \]

    Rearranging to make $\rho$ the subject, we see that $\rho$ must obey the following:

    \[
    \rho > \frac{r_2^2}{2r_1^2} \sigma_1^2
    \]

    If this condition holds, then $Var( \cnobold_i - \chatnobold_i) < Var(\cnobold_i - \mu_i)$ as required.

\end{proof}

\subsection{Discussion}
\label{sec:statistic_discussion}
Initially, it is obvious that:

\[
\mathbb{E}||\chat - \mathbf{c}||^2_2 < \mathbb{E}||\boldsymbol{\mu} - \mathbf{c} ||^2_2
\]

due to the variance result proved above. However, we wish to make some statements about the \emph{normalised} vectors $\normc$
, $\normchat$, and $\boldsymbol{\mu}$.

By normalising these vectors, we preserve only the angular component of the distance and discard any radial component. For example if $\chat$, $\mathbf{c}$, and $\normc$ are perfectly co-linear but have different $\ell_2$-norms, then there will be no angular error at all once all vectors are normalised, despite possibly large distances between $\mathbf{c}$ and $\chat$ in Euclidean space. Conversely, if $||\chat||_2$ or $||\mathbf{c}||_2$ is small, even small \emph{angular} changes betweem these vectors could result in large changes to the normalised distance $||\normc - \normchat||_2$.

Thus, showing $\mathbb{E}||\chat - \mathbf{c}||_2 < \mathbb{E}||\boldsymbol{\mu} - \mathbf{c} ||_2$ does not directly translate to distances between the normalised vectors. However, under the assumptions of our model, there is no reason to assume that error will be concentrated in the angular direction. Thus, $\normchat$ should retain its advantage over $\boldsymbol{\mu}$ unless adversarial conditions are in effect.

\subsection{Lower Bound on $\rho$}
To place a lower bound on $\rho$, we first assume that at least some fraction $\alpha k$ of the candidate set consists of true near neighbours. We therefore define a partition on $S$:

\[
    S = I \cup O
\]

where $I$ is the set of true near neighbours, whereby $|I| \geq \alpha k$ and $O$ is the set of non-near neighbours whereby $|O| < k(1 - \alpha)$. We now show that $\rho$ can be expressed in terms of $\alpha$ and $\sigma_1^2$:

\begin{align*}
    \rho &= Cov(\chatnobold_i - \mu_i, c_i - \mu_i)\\
    &= Cov \left( \frac{1}{k} \left[\sum_{\mathbf{y} \in I} y_i +  \sum_{\mathbf{z} \in O} z_i \right], \frac{1}{k}\sum_{\mathbf{x} \in \knnmu} x_i \right) \\
    &= \frac{1}{k^2}Cov \left(\left[\sum_{\mathbf{y} \in I} y_i +  \sum_{\mathbf{z} \in O} z_i \right],\sum_{\mathbf{x} \in \knnmu} x_i \right) \\
    &\texttt{By addition rule}:\\
    &= \frac{1}{k^2} \left( Cov \left( \sum_{\mathbf{y} \in I} y_i, \sum_{\mathbf{x} \in \knnmu} x_i \right) + \left( Cov \left( \sum_{\mathbf{z} \in O} z_i, \sum_{\mathbf{x} \in \knnmu} x_i \right) \right) \right)
\end{align*}

Here, we assume that the covariance between the sets $O$ and $\knnmu$ are at worst $0$. Under adversarial conditions it is a strong possibility this assumption will fail. However, for the embedding datasets we study, it is unlikely that a set of false positives with a negative covariance to $\knnq$ could be retrieved.

\begin{align*}
    &= \frac{1}{k^2} Cov \left( \sum_{\mathbf{y} \in I} y_i, \sum_{\mathbf{x} \in \knnmu} x_i \right)\\
    &= \frac{1}{k^2} Cov \left( \sum_{\mathbf{y} \in I} y_i, \sum_{\mathbf{y} \in I} y_i + \sum_{\mathbf{w} \in \knnmu \setminus I} w_i \right)\\
    &\texttt{By addition rule}:\\
    &= \frac{1}{k^2} \left( Cov \left( \sum_{\mathbf{y} \in I} y_i, \sum_{\mathbf{y} \in I} y_i \right) + Cov \left( \sum_{\mathbf{y} \in I} y_i, \sum_{\mathbf{w} \in \knnmu \setminus I} w_i \right)  \right)
\end{align*}

Here, we once again assume that members of the set $I$ and $\knnmu \setminus I$ have at worst $0$ correlation.

\begin{align*}
    &\geq \frac{1}{k^2} \left( Cov \left( \sum_{\mathbf{y} \in I} y_i, \sum_{\mathbf{y} \in I} y_i \right) \right) \\
    &\texttt{By variance identity} \\
    &\geq \frac{1}{k^2} Var  \left( \sum_{\mathbf{y} \in I} y_i \right) \\
    &\geq \frac{\alpha k \xi_1^2}{k^2} \\
    \rho&\geq \alpha \sigma^2_1
\end{align*}

As shown at the end of Section \ref{sec:var_diff}, we require:

\begin{align*}
    \rho &> \frac{r_2^2}{2r_1^2} \sigma_1^2 \\
\end{align*}

Therefore,

\begin{align*}
    \alpha \sigma^2_1 &> \frac{r_2^2}{2r_1^2} \sigma_1^2 \\
    \alpha &> \frac{r_2^2}{2r_1^2}
\end{align*}

\subsubsection{Discussion}
In a high-dimensional spaces we investigate, $r_1 \approx r_2$. To show this, we measure the mean value of $r_1$ at $k=100$, and $r_2$ at $k=500$ for $250$ queries, and show the results in Table \ref{tab:r-values}.

\begin{table}[h]
\centering
\begin{tabular}{lccc}
\toprule
& Glove & MirFlickr & GOOAQ \\
\midrule
$r_1$ & 0.956 & 0.944 & 0.787 \\
$r_2$ & 1.007 & 1.031 & 0.919 \\
$\alpha =\frac{r_2^2}{2r_1^2}$ & 0.555 & 0.596 & 0.681\\
\bottomrule
\end{tabular}
\caption{Values of $r_1$ and $r_2$ for each dataset, measured over $250$ randomly drawn queries.}
\label{tab:r-values}
\end{table}

In the worst case, the required $\alpha$ value for GooAQ is $0.68$. However, if we can assume that at least $\alpha k$ near neighbours are correct, we can improve this bound as we demonstrate in Section \ref{sec:tighter_bound}.

\subsection{A Tighter Bound}
\label{sec:tighter_bound}
If we can assume that at least a fraction $\alpha$ of elements in S are true near neighbours, then we can reason that:

\[
\sigma_2^2 = \frac{\alpha k}{k}\sigma_1^2 + \frac{r_2^2(k-\alpha k)}{r_1^2k} \sigma_1^2
\]

as at least a fraction $\alpha$ summands have variance $\sigma_1^2$. Simplified, this becomes:

\[
\sigma_2^2
= \sigma_1^2\left(\alpha+\frac{r_2^2(1-\alpha)}{r_1^2}\right)
\]

Substituting into Equation \ref{eq:c_ci_dist}:

\[
(\cnobold_i - \chatnobold_i) \sim \mathcal{N}(0, \sigma^2_1 + \sigma_1^2\left(\alpha+\frac{r_2^2(1-\alpha)}{r_1^2}\right) - 2 \rho)
\]

As before, we require the variance to be less than $\sigma_1^2$:

\[
\sigma^2_1 + \sigma_1^2\left(\alpha+\frac{r_2^2(1-\alpha)}{r_1^2}\right) - 2 \rho < \sigma_1^2
\]

Rearranged, this becomes:

\[
\rho > \frac{\sigma_1^2}{2}\left(\alpha+\frac{r_2^2(1-\alpha)}{r_1^2}\right)
\]

Substituting $\rho$ with $\alpha \sigma_1^2$:

\[
\alpha \sigma_1^2 \geq \frac{\sigma_1^2}{2}\left(\alpha+\frac{r_2^2(1-\alpha)}{r_1^2}\right)
\]

so,

\[
\alpha \geq \frac{1}{2}\left(\alpha+\frac{r_2^2(1-\alpha)}{r_1^2}\right)
\]

Rearranging:

\[
\alpha \geq \frac{r_2^2}{r_1^2 + r_2^2}
\]

Thus, using this new bound, the required values of $\alpha$, analogous to Table \ref{tab:r-values} are:

\begin{table}[h]
\centering
\begin{tabular}{lccc}
\toprule
& Glove & MirFlickr & GOOAQ \\
\midrule
$r_1$ & 0.956 & 0.944 & 0.787 \\
$r_2$ & 1.007 & 1.031 & 0.919 \\
$\alpha = \frac{r_2^2}{r_1^2 + r_2^2}$ & 0.526 & 0.544 & 0.577\\
\bottomrule
\end{tabular}
\caption{Values of $r_1$ and $r_2$, and the tighter bound for $\alpha$ for each dataset, measured over $250$ randomly drawn queries.}
\label{tab:r-values-2}
\end{table}





\section{Frequency of Candidate Set Updates}
\label{sec:update_frequency}
Here we measure how often RP-Forest adds points to its candidate set after the $i$th tree is searched, thereby estimating the rate at which the candidate set changes. If this value is below $\frac{k}{2}$, then the fast query centroid algorithm set out in Section \ref{sec:fast_centroid_cal} can be used to decrease the cost of computing $\normchat$.

To measure this property, we formulate an experiment whereby an RP-Forest formed of $19$ trees is constructed from each dataset. We then search the constructed forests using $100$ queries drawn and excluded from each dataset at random, before each forest was constructed. After each tree in the forest is searched, we measure the number of elements which were added to the candidate set. We report this metric as $\Delta-kNN$. We show the average $\Delta-kNN$ value after each tree is searched for all datasets in Figure \ref{fig:update_rate}.

We show that, on average, points are added to the candidate set infrequently. Critically, we show that no dataset adds points to the candidate set at a rate of above $\frac{k}{2}$ per tree beyond the first. Therefore the fast centroid computation algorithm (described in Section \ref{sec:fast_centroid_cal}) can be used to decrease the cost of computing $\normchat$.

\begin{figure}[h]
    \centering
    \begin{subfigure}{0.48\textwidth}
        \centering
        \includegraphics[width=\linewidth]{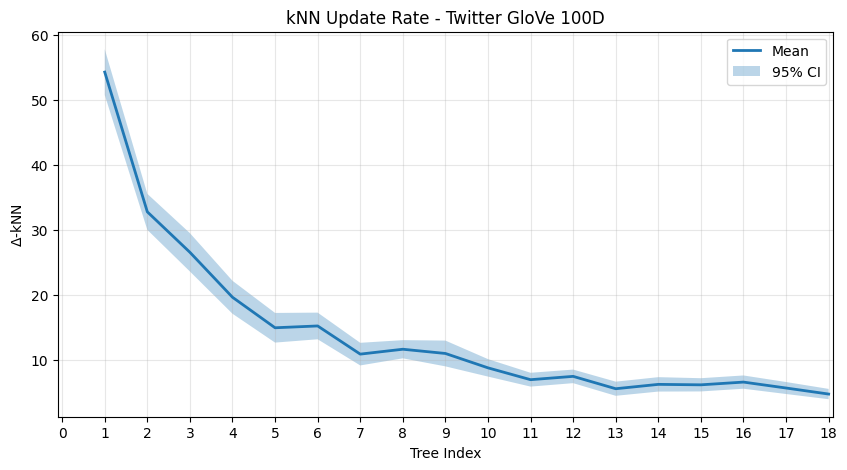}
        \caption{GloVe}
        \label{fig:update_rate_glove}
    \end{subfigure}
    \hfill
    \begin{subfigure}{0.48\textwidth}
        \centering
        \includegraphics[width=\linewidth]{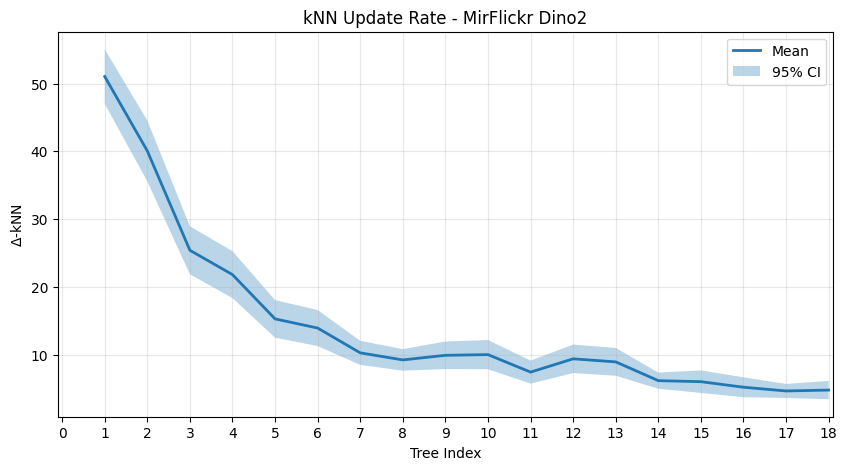}
        \caption{DINOv2}
        \label{fig:update_rate_dino2}
    \end{subfigure}

    \par\medskip

    \begin{subfigure}{0.48\textwidth}
        \centering
        \includegraphics[width=\linewidth]{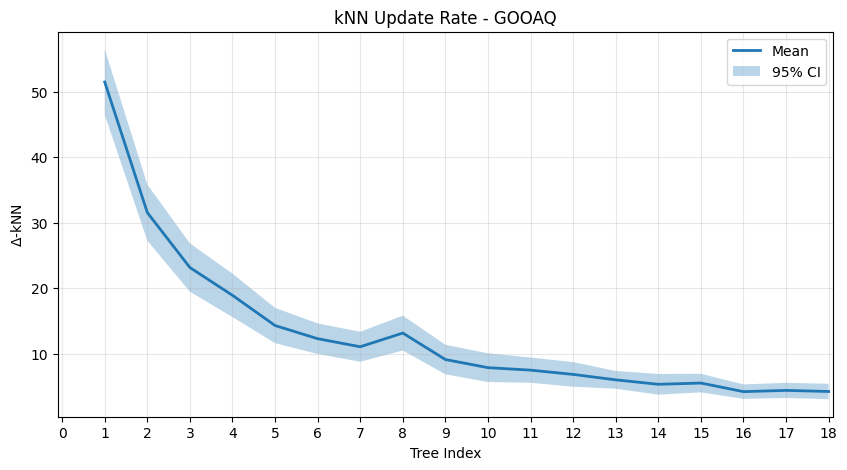}
        \caption{GooAQ}
        \label{fig:update_rate_gooaq}
    \end{subfigure}
    \hfill
    \begin{subfigure}{0.48\textwidth}
        \centering
        \includegraphics[width=\linewidth]{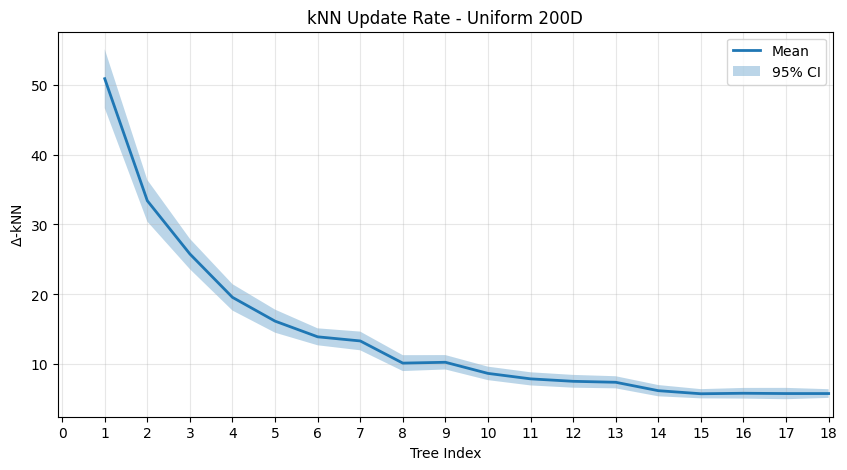}
        \caption{Uniform 200D}
        \label{fig:update_rate_uniform}
    \end{subfigure}

    \caption{Rate of change of the candidate set after trees in an RP-Forest are searched. $k=100$.}
    \label{fig:update_rate}
\end{figure}

\end{document}